\documentclass{amsart}

\newcommand*{\mailto}[1]{\href{mailto:#1}{\nolinkurl{#1}}}
\usepackage{hyperref}

\usepackage{amsfonts,amsmath,amsxtra,amsthm,amssymb,latexsym}
\usepackage{color}

\theoremstyle{plain}
\newtheorem{theorem}{Theorem}[section]
\newtheorem{corollary}[theorem]{Corollary}
\newtheorem{lemma}[theorem]{Lemma}

\newtheorem{definition}{Definition}[section]
\newtheorem{remark}[theorem]{Remark}
\newtheorem{example}[theorem]{Example}
\newtheorem{proposition}[theorem]{Proposition}

\newcommand{\floor}[1]{\lfloor#1 \rfloor}
\newcommand{\nn}{\nonumber}
\newcommand{\be}{\begin{equation}}
\newcommand{\ee}{\end{equation}}

\numberwithin{equation}{section}

 \DeclareMathOperator{\im}{Im}
 
 \DeclareMathOperator{\dom}{dom}
\DeclareMathOperator{\ran}{ran} 
\DeclareMathOperator{\Span}{span}\DeclareMathOperator{\ess}{ess}

\DeclareMathOperator{\diag}{diag}\DeclareMathOperator{\loc}{loc}
\DeclareMathOperator{\comp}{comp}

\newcommand\R{{\mathbb{R}}}
\newcommand\C{{\mathbb{C}}}

\newcommand\N{{\mathbb{N}}}

\newcommand\gH{{\mathfrak{H}}}
\newcommand\gN{{\mathfrak{N}}}
\newcommand{\gG}{{\Gamma}}

\newcommand{\gt}{\mathfrak{t}}

\newcommand{\gd}{{d}}
\newcommand{\gA}{{\alpha}}

\newcommand\cH{{\mathcal{H}}}
\newcommand\cI{{\mathcal{I}}}

\newcommand\cK{{\mathcal{K}}}

\newcommand\rH{{\mathbf{H}}}
\newcommand\rh{{\mathbf{h}}}

\newcommand\rD{{\rm{d}}}
\newcommand\I{{\rm{i}}}

\begin{document}

\title[Concentric $\delta$-shells]{Schr\"odinger operators\\ with concentric $\delta$-shells}

\author[S.\ Albeverio]{Sergio Albeverio}
\address{Institut f\"ur Angewandte Mathematik\\ HCM\\ IZKS\\ SFB611\\
Universit\"at Bonn\\ Endenicher  Allee 60\\
53115 Bonn\\ Germany\\ and CERFIM\\ Locarno\\ Switzerland}
\email{\mailto{albeverio@uni-bonn.de}}

\author[A.\ Kostenko]{Aleksey Kostenko}
\address{Fakult\"at f\"ur Mathematik\\
Universit\"at Wien\\
Nordbergstr. 15\\
1090 Wien, Austria}
\email{\mailto{Oleksiy.Kostenko@univie.ac.at; duzer80@gmail.com}}

\author[M.\ Malamud]{Mark Malamud}
\address{Institute of Applied Mathematics and Mechanics\\
NAS of Ukraine\\ R. Luxemburg str. 74\\
Donetsk 83114\\ Ukraine}
\email{\mailto{mmm@telenet.dn.ua}}

\author[H.\ Neidhardt]{Hagen Neidhardt}
\address{Institut f\"ur Angewandte Analysis und Stochastik\\
Mohrenstr. 39\\
D-10117 Berlin\\
Germany}
\email{\mailto{neidhard@wias-berlin.de}}

\thanks{{\it The research was funded by DFG under project No.\ 436 UKR 113/85/0-1 and by the Austrian Science Fund (FWF) under project No.\ M1309--N13}}




\begin{abstract}
We investigate the spectral properties of the Schr\"odinger operators in $L^2(\R^n)$ with a singular interaction supported by an infinite family of concentric spheres
\[
\rH_{R,\gA}=-\Delta+\sum_{k=1}^\infty\gA_k\delta(|x|-r_k).
\]
We obtain necessary and sufficient conditions for the operator $\rH_{R,\gA}$ to be self-adjoint, lower-semibounded. Also we investigate the spectral types of $\rH_{R,\gA}$.
\end{abstract}

\maketitle

\section{ Introduction}\label{intro}

We analyze the spectral properties of the Schr\"odinger operators in $L^2(\R^n)$, $n\ge 2$, with a singular interaction supported by an infinite family of concentric spheres, analogous to a system studied by Exner and Fraas \cite{ef_07, ef_08},
\begin{equation}\label{I_01}
\rH_{R,\gA}=-\Delta+\sum_{k=1}^\infty\gA_k\delta(|x|-r_k).
\end{equation}
Note that in \cite{ef_07,ef_08} the case of radially periodic interactions,  $\gA_k\equiv\gA$ and $r_k=r_0+Tk$, has been considered. Motivated by the paper of Hempel, Hinz, and Kalf \cite{hhk_87}, Exner and Fraas in \cite{ef_07} gave a complete characterization of the spectrum of the Hamiltonian $\rH_{R,\gA}$ with radially periodic interactions.

The $\delta$ sphere interaction, formally given in three dimensions by the Hamiltonian
$\rH = -\Delta + a\delta(|x|- r_0)$, has a long history. The physical motivation was coming mainly from nuclear physics, where the model was introduced by Green and Moszkowski \cite{gm_65} under the name of {surface delta interaction}. Other applications may be found in molecular \cite{bli_78} and solid state physics \cite{llo_65, rgm_67}. A rigorous mathematical treatment of the $\delta$ sphere interaction was first given in \cite{ags_87} (see also \cite{sha_88} for the case of  finitely many $\delta$ sphere interactions).

%
%

In the present paper we are interested in spectral properties of the operator $\rH_{R,\gA}$ in the case of arbitrary positions $r_k$ and strengths $\gA_k$ of the interactions. We are going to study the following problems: self-adjointness, lower semiboundedness, characterization of the spectrum.

Namely, assume that the sequence of radii $R=\{r_k\}_{k=1}^\infty$ is strictly increasing and can accumulate only at infinity, $r_k\uparrow +\infty$.
Since the potential is spherically symmetric, it is natural to use a
partial wave decomposition (see, e.g., \cite{ags_87, sha_88}). Using the isometry
\begin{equation}\label{eq:31.01}
    \mathsf{U}\,:\,L^2((0,\,\infty),r^{n-1}dr)\, \rightarrow
    \, L^2(0,\,\infty), \quad \mathsf{U}f(r) =
    r^{\frac{n-1}{2}}f(r),
\end{equation}
we get
\begin{equation}\label{eq:31.02}
    L^2(\R^n) =
    L^2((0,\,\infty),r^{n-1}dr)\otimes L^2(S^{n-1})=\bigoplus_{l=0}^\infty\, \mathsf{U}^{-1}L^2(0,\,\infty)\otimes \mathcal{H}_l
\end{equation}
and
\begin{equation}\label{eq:31.03}
   \rH_{R,\gA} = \bigoplus_{l=0}^\infty\, \mathsf{U}^{-1}
   \rh_{R,\gA}^{(l)}\mathsf{U} \otimes \varkappa_lI_l .
\end{equation}
Here $\mathcal{H}_l$ is the eigenspace corresponding to the $l$-th eigenvalue $\varkappa_l=-l(l+n-2)$ of the Laplace--Beltrami operator on $L^2(S^{n-1})$ and $I_l$ is the unit operator on $\cH_l$. The operator
    \begin{align}\label{III_02}
    \rh_{R,\gA}^{(l)}= -\frac{\rD^2}{\rD r^2} + \Big(\frac{(n-1)(n-3)}{4} -\varkappa_l\Big)\frac{1}{r^2}
    +  \sum_{k=1}^\infty \gA_k \delta(r-r_k). 
    \end{align}
is defined as the closure of the following minimal symmetric operator
 \begin{align}\label{III_02B}
 \rh_{R,\gA}^{(l)}&=\overline{\rh_{\min}^{(l)}},\qquad \rh_{\min}^{(l)}f=\tau^{(l)}[f],\\
  \tau^{(l)}&=-\frac{\rD^2}{\rD r^2} + \Big(\frac{(n-1)(n-3)}{4} -\varkappa_l\Big)\frac{1}{r^2},
  \end{align}
  \begin{align}
\dom(\rh_{\min}^{(l)})=\{f\in L^2_{\comp}(\R_+)&: f\in  C^\infty(\R_+\setminus R),\ \tau^{(l)}[f]\in L^2(\R_+),\nonumber\\
 &\  \begin{array}{c}f(r_k+)=f(r_k-),\\
 f'(r_k+)-f'(r_k-)=\gA_k f(r_k),
 \end{array}\ k\in\N \}.\nonumber
    \end{align}
If $l(n):=\frac{(n-1)(n-3)}{4} -\varkappa_l\in [-1/2,1/2)$, then the functions from $\dom(\rh_{\min}^{(l)})$ are assumed to satisfy the following boundary condition at $r=0$\footnote{Note that in \cite{bg85} and \cite{ef_07} a somewhat different boundary condition $\lim_{r\to 0} \frac{f(r)}{\sqrt{r}\log r}=0$ is used in the case $n=2$ and $l=0$, i.e., when $l(n)=-1/2$. In fact, these conditions coincide since in this case the minimal symmetric operator associated with $-\frac{d^2}{dr^2}-\frac{1}{4r^2}$ in $L^2(\R_+)$ has a unique positive self-adjoint extension.}
\be\label{eq:1.8}
\lim_{r\to 0}r^{l(n)}((l(n)+1)f(r)-rf'(r))=0.
\ee

So, the spectral analysis of the operator $\rH_{R,\gA}$ is clearly reduced to the analysis of the Bessel operators $\rh_{R,\gA}^{(l)}$ with local point interactions. Our main aim is to apply the results obtained in the recent papers \cite{AKM_10, KM_09} for studying the properties of the Hamiltonians \eqref{I_01}.

The plan of the paper is as follows. Section \ref{Sec_II_Prelim} is of preliminary character. Here we collect the results on spectral properties of  Sturm--Liouville operators $\rh_{R,\gA,q}^{(0)}$ with $\delta$-interactions. In Section \ref{sec:iii}, we show that all the results from Section \ref{Sec_II_Prelim} can be extended to the case of an arbitrary $l\ge -\frac{1}{2}$. In Section \ref{sec:bstates}, we collect results on the number of negative squares (eigenvalues) of the operators $\rh_{R,\gA}^{(l)}$. In particular, we prove the analog of the classical Bargmann estimate for operators with $\delta$-interactions. Note that this problem has attracted some attention   recently (see \cite{AlbNiz03, GolOr_10, KM_2_09, Ogu08, Ogu10} and references therein). In the final Section \ref{sec:v} we describe the main spectral properties of multi-dimensional Schr\"odinger operators $\rH_{R,\gA}$ with concentric $\delta$-shells. In Appendix we collect necessary notions and facts on the concept of boundary triplets and Weyl functions.


{\bf Notation.} $\N, \C, \R$ have the usual meaning; $\R_+=[0,\infty)$.

$\dom (T)$, $\ker (T)$, $\ran (T)$ are the domain, the kernel,
the range of a linear operator $T$ in a Hilbert space $\mathfrak{H}$, respectively;
$ R_T \left(\lambda \right):=\left( T-\lambda I\right)^{-1} $, $\lambda \in \rho(T)$,
is the resolvent of $T$; $\sigma(T)$
and $\rho(T)$ denote the spectrum and the resolvent set of $T$.

$E_T(\cdot)$ denotes the spectral measure of a  self--adjoint
operator $T = T^*$ in $\gH$, $T_-:= TE_T(-\infty, 0)$ and $T_+:= TE_T(0,+\infty)$ are the negative and positive parts of the operator $T$, respectively, and $\kappa_\pm(T):=\dim\big(\ran(T_\pm)\big)$ (if $\kappa_\pm(T)<\infty$, then $\kappa_\pm(T)$ is  the number of negative/positive eigenvalues of $T$ counting multiplicities).

Let $R$ be a discrete subset of $\R_+$, $R=\{r_k\}_1^\infty$ and $r_k\uparrow +\infty$. Then $C^\infty(\R_+\setminus R)$ is the set of infinitely differentiable functions on each interval $[r_{k-1},r_k]$.
Also we shall use the following Sobolev spaces ($p\ge 1$)
\begin{eqnarray*}
&W^{2,p}_0(\R_+\setminus R):=\{f\in
W^{2,p}(\R_+): f(r_k)=f'(r_k)=0,\, k\in \N\},\\
&W^{2,p}(\R_+\setminus R):=\{f\in L^2(\R_+): f\in W^{2,p}[r_{k-1},r_k], k\in \N,\,  f''\in
  L^p(\R_+)\},\\
&W^{2,p}_{\comp}(\R_+\setminus R):=W^{2,p}(\R_+\setminus R)\cap L^p_{\comp}(\R_+).
\end{eqnarray*}


\section{1-D Schr\"odinger operators with $\delta$-interactions}\label{Sec_II_Prelim}


Let us first briefly recall the main properties of one-dimensional systems with $\delta$ interactions \cite{Alb_Ges_88, AKM_10, KM_09}.
We consider the differential expression
\[
\tau_{R,\alpha,q} = -\frac{\rD^2}{\rD r^2} + q(r) + \sum_{k \in \mathbb{N}}\alpha_k \delta(r - r_k),\quad r,r_k>0,
\]
with $\delta$-type interactions at points $r_k$ accumulating only at $\infty$, $r_k\uparrow +\infty$.
Namely, define the operator
\begin{equation}\label{I_02}
\rh_{R,\gA,q}'=\tau_{q}:=-\frac{\rD^2}{\rD r^2}+q(r),\quad q\in L^1_{\loc}(\R_+),
  \end{equation}
on the minimal domain
   \begin{align}\label{I_02B}
\dom(\rh_{R,\gA,q}')=\{f\in W_{\comp}^{2,1}(\R_+\setminus R):\, 
 f(0)=0,\ f(r_k+)=f(r_k-),\\
 f'(r_k+)-f'(r_k-)=\gA_k f(r_k),\ \tau_qf\in L^2(\R_+) \}.\nonumber
    \end{align}
It is clear that $\rh_{R,\gA,q}'$ is a symmetric operator. Let
$\rh_{R,\gA,q}$ denote the closure of $\rh_{R,\gA,q}'$ in
$L^2(\R_+)$,
    \begin{equation}\label{I_02C}
\rh_{R,\gA,q} = \overline{\rh_{R,\gA,q}'}.
   \end{equation}
   If $q\equiv 0$, we set $\rh_{R,\gA}:=\rh_{R,\gA,0}$.
   
   In the recent paper  \cite{KM_09}, two of us investigated the
Hamiltonian $\rh_{R,\gA,q}$  in the framework of the extension
theory of symmetric operators. More precisely, 
applying  the technique of boundary triplets and the corresponding
Weyl functions (see, e.g., \cite{Gor84, DM91}), it is shown in \cite[\S 5]{KM_09}
that 
self-adjontness, lower semiboundedness, and  discreteness of the spectrum of $\rH_{R,\gA}$
correlate with the corresponding spectral properties of the  Jacobi matrix
       \begin{equation}\label{I_07}
B_{R,\gA}=\left(\begin{array}{cccc}
p_1^{-2}\bigl(\alpha_1+\frac{1}{\gd_1}+\frac{1}{\gd_2}\bigr) & -(p_1p_2\gd_2)^{-1} & 0&   \dots\\
(p_1p_2\gd_2)^{-1} &p_2^{-2}\bigl(\alpha_2+\frac{1}{\gd_2}+\frac{1}{\gd_3}\bigr) & -(p_2p_3\gd_3)^{-1} &  \dots\\
0 & -(p_2p_3\gd_3)^{-1} & p_3^{-2}\bigl(\alpha_3+\frac{1}{\gd_3}+\frac{1}{\gd_4}\bigr)&  \dots\\
\dots & \dots& \dots & \dots
\end{array}\right),
    \end{equation}
where
\be\label{I_07B}
\gd_k:=r_k-r_{k-1},\quad \text{and} \quad p_k:=\sqrt{\gd_k+\gd_{k+1}},\quad k\in\N.
\ee
Namely, with $B_{R,\gA}$ one associates in $l^2$ a closed minimal symmetric operator, also denoted by $B_{R,\gA}$ (cf. \cite{Ber68}). Then the following theorem holds.
%
\begin{theorem}[\cite{KM_09, KM_2_09}]\label{th_KM}
Let $\rh_{R,\gA,q}$ and $B_{R,\gA}$ be the minimal symmetric operators defined by \eqref{I_02}--\eqref{I_02C} and \eqref{I_07}--\eqref{I_07B}, respectively. Assume also that $q\in L^\infty(\R_+)$ and $\gd^*:=\sup_{k}\gd_k<\infty$. Then:
\item $(i)$ \ \ $\rh_{R,\gA,q}$ is self-adjoint if and only if so is $B_{R,\gA}$. Moreover, $n_\pm(\rh_{R,\gA,q})=n_\pm(B_{R,\gA})\le 1$.
\item $(ii)$ \ \ $\rh_{R,\gA,q}$ is lower semibounded (non-negative) if and only if so is $B_{R,\gA}$.
\item $(iii)$ \ \ If $\rh_{R,\gA,q}=\rh_{R,\gA,q}^*$, then $\sigma(\rh_{R,\gA,q})$
is discrete precisely when $\gd_k\to 0$ and  $\sigma(B_{R,\gA,q})$ is discrete.
\item $(iv)$ \ \ Assume in addition that $q\equiv 0$. 
Then the negative spectrum of $\rh_{R,\gA}$
is discrete (finite) 
if and only if the negative spectrum of $B_{R,\gA}$ is also discrete (finite). Moreover,
$\kappa_-(\rh_{R,\gA})= \kappa_-(B_{R,\gA})$.
\end{theorem}

Using the form approach developed in \cite{AKM_10} for lower semibounded Hamiltonians with $\delta$-interactions, Theorem \ref{th_KM} can be made more specific and detailed in several directions.
\begin{theorem}\label{th_III.1}
If the minimal operator $\rh_{R,\alpha,q}$ is lower semibounded, then it is self-adjoint,
$\rh_{R,\gA,q} = (\rh_{R,\gA,q})^*$.

In particular, if 
\begin{equation}\label{I_brinck_q}
\sup_{r>0}\int_r^{r+1} |q_-(t)|dt <+\infty,\qquad q_\pm(r):=(q(r)\pm |q(r)|)/2,
\end{equation}
and
\begin{equation}\label{I_brinck}
\sup_{r>0}\sum_{r_k\in[r,r+1]}|\gA_k^-| <+\infty,\qquad \gA_k^\pm:=(\gA_k\pm |\gA_k|)/2,
\end{equation}
then the operator $\rh_{R,\gA,q}$ is lower semibounded and hence is self-adjoint.
\end{theorem}
Next  the classical  Molchanov discreteness
criterion has been extended in \cite{AKM_10} to the case of the Hamiltonians $\rh_{R,\gA,q}$.   
      \begin{theorem}\label{th_discretcriter}
Let  the potential $q$ and the sequence $\gA$ satisfy \eqref{I_brinck_q} and \eqref{I_brinck}, respectively. 
The spectrum $\sigma(\rh_{R,\alpha,q})$ of the operator
$\rh_{R,\alpha,q}$ is discrete if and only if for every
$\varepsilon>0$
    \begin{equation}\label{2.1}
\int_r^{r+\varepsilon} q(t)dt+\sum_{r_k\in(r,r+\varepsilon)}\alpha_k \to \infty \quad \text{as}\quad r\to
\infty.
    \end{equation}
           \end{theorem}
Moreover, the version of Birman's result \cite{Bir61} on
stability of a continuous spectrum holds true.
     \begin{theorem}[\cite{AKM_10}]\label{thContSpec}
         Let $q$ and $\gA$ satisfy \eqref{I_brinck_q} and \eqref{I_brinck}, respectively. If
                  \be
         \lim_{r\to 0}\int_r^{r+1} |q(t)|dt=0
         \ee
         and
               \begin{equation}\label{3.2}
\lim_{r\to\infty}\sum_{r_k\in[r,r+1]}|\alpha_k|=0,
     \end{equation}
          then $\sigma_c(\rh_{R,\gA,q})=\sigma_c(\rh_{q})=\R_+$.
           \end{theorem}

It is interesting to note  that the condition
              \begin{equation}\label{3.2C}
 \lim_{k\to\infty}|\alpha_k|=0,
           \end{equation}
is insufficient for the Hamiltonian $\rh_{R,\gA}$ to have a continuous
spectrum (see \cite{KM_09, AKM_10}). Moreover, it may even happen that \eqref{3.2C}
is satisfied, although the spectrum $\sigma(\rh_{R,\gA})$ is purely
discrete (see \cite[Remark 4.6]{AKM_10}).

\begin{remark}\label{rem:meas}
Let us mention that Theorems \ref{th_III.1}, \ref{th_discretcriter}, and \ref{thContSpec} remain true for potentials that are locally finite measures on $\R_+$ (see \cite{AKM_10}).
\end{remark}

\section{Bessel operators with local point interactions}\label{sec:iii}

The main goal of this section is to extend the results of Section \ref{Sec_II_Prelim} to the case of Bessel type operators. Namely, assume that $r_k\uparrow +\infty$, $\gA=\{\gA_k\}_1^\infty\subset\R$, $q\in L^1_{\loc}(\R_+)$, and $l\ge-1/2$. Consider the following differential expression in $L^2(\R_+)$
\be\label{eq:30.01}
\tau_{R,\gA,q}^{(l)}= -\frac{\rD^2}{\rD r^2} + \frac{l(l+1)}{r^2}+q(r)
    +  \sum_k \gA_k \delta(r-r_k).
\ee
If $R=\emptyset$ or $\alpha\equiv 0$, then $\tau_{q}^{(l)}:=\tau_{R,0,q}^{(l)}=\tau_{\emptyset,\gA,q}^{(l)}$ and $\tau^{(l)}:=\tau_0^{(l)}$.
Define the operator
\be\label{eq:30.02}
{(\rh_{R,\gA,q}^{(l)})'}f:=\tau_{q}^{(l)}[f],\qquad  \tau_{q}^{(l)}:=-\frac{\rD^2}{\rD r^2} + \frac{l(l+1)}{r^2}+q(r)
\ee
 on the minimal domain 
\begin{align}
\dom\big((\rh_{R,\gA,q}^{(l)})'\big)=&\big\{f\in W^{2,1}_{\comp}(\R_+\setminus R): \,  f(r_k+)=f(r_k-),\\
 &f'(r_k+)-f'(r_k-)=\gA_k f(r_k),\, 
 \tau_q^{(l)}f\in L^2(\R_+)
 \big\}.\label{eq:30.03}
\end{align}
If $l\in [-\frac{1}{2},\frac{1}{2})$, then we also impose the usual boundary conditions at $r=0$,
\be\label{eq:bc0}
\lim_{r\to0}r^l\big((l+1)f(r)-rf'(r)\big)=0.
\ee
Clearly, the operator $(\rh_{R,\gA,q}^{(l)})'$ is symmetric. Let us denote its closure by $\rh_{R,\gA,q}^{(l)}$,
\be\label{eq:30.05a}
\rh_{R,\gA,q}^{(l)}:=\overline{(\rh_{R,\gA,q}^{(l)})'}
\ee

\subsection{Connection with Jacobi matrices}\label{ss:iii_1}

In this subsection we are going to establish the analog of Theorem \ref{th_KM}.

\begin{lemma}\label{lem:_KM}
Let $\rh_{R,\gA,q}^{(l)}$ and $B_{R,\gA}$ be the minimal symmetric operators defined by \eqref{eq:30.01}--\eqref{eq:30.05a} and \eqref{I_07}--\eqref{I_07B}, respectively. Assume also that $q\in L^\infty(\R_+)$ and $\gd^*:=\sup_{k}\gd_k<\infty$. Then:
\item $(i)$ \ \ $\rh_{R,\gA,q}^{(l)}$ is self-adjoint if and only if so is $B_{R,\gA}$. Moreover, $n_\pm(\rh_{R,\gA,q}^{(l)})=n_\pm(B_{R,\gA})\le 1$.
\item $(ii)$ \ \ $\rh_{R,\gA,q}^{(l)}$ is lower semibounded if and only if so is $B_{R,\gA}$. Moreover, $\rh_{R,\gA,q}^{(l)}$ is nonnegative whenever so are $l$ and $B_{R,\gA}$.
\item $(iii)$ \ \ If $\rh_{R,\gA,q}^{(l)}=\big(\rh_{R,\gA,q}^{(l)}\big)^*$, then $\sigma(\rh_{R,\gA,q}^{(l)})$
is discrete precisely when $\gd_k\to 0$ and  $\sigma(B_{R,\gA})$ is discrete.
\item $(iv)$ \ \ Assume that $q\equiv 0$. The negative spectrum of $\rh_{R,\gA}^{(l)}$
is discrete  
if and only if the negative spectrum of $B_{R,\gA}$ is also discrete.
\item $(v)$ \ \ Assume that $q\equiv 0$ and $l\ge 0$. The negative spectrum of $\rh_{R,\gA}^{(l)}$
is  finite 
if so is the negative spectrum of $B_{R,\gA}$. 
Moreover,
$\kappa_-(\rh_{R,\gA}^{(l)}) \le  \kappa_-(B_{R,\gA})$.
\end{lemma}
\begin{proof}
Choose $c\in (0,r_1)$. Then the operator $\rh_{R,\gA,q}^{(l)}$ is a rank one perturbation (in the resolvent sense) of the following direct sum operator
\[
\rh^{(l),D}_{R,\gA,q}(0,c)\oplus \rh^{(l),D}_{R,\gA,q}(c,\infty).
\]
Here $\rh_{R,\gA,q}^{(l),D}(\cI)$ denotes the operator obtained by restricting $\rh_{R,\gA,q}^{(l)}$ to the interval $\cI\in\{(0,c),(c,+\infty)\}$ and subject to the Dirichlet boundary conditions at $r=c$.

Further, we observe that the potential
\[
q^{(l)}(r)=\frac{l(l+1)}{r^2}
\]
tends to $0$ as $r\to \infty$ and also is bounded on $(c,+\infty)$ for all $l\ge -1/2$. 
Therefore, Theorem \ref{th_KM} clearly holds for the operator $\rh^{(l),D}_{R,\gA,q}(c,\infty)$.

Again, since $l\ge-1/2$, $q\in L^\infty(0,c)$ and $c<r_1$, the operator $\rh^{(l),D}_{R,\gA,q}(0,c)=\rh^{(l),D}_{q}(0,c)$ is self-adjoint, lower semibounded and has purely discrete spectrum. The proof of the lemma is completed by using the fact that  self-adjointness, lower semiboundedness, discreteness, as well as a continuous spectrum are stable under finite rank perturbations. Moreover, the same arguments prove statements $(iv)$ and $(v)$.
\end{proof}
\begin{remark}
Let us note that condition $(v)$ is only sufficient since finiteness of the negative spectrum of $\rh_{R,\gA}^{(l)}$ is not stable under perturbations by critical potentials $q(r)=\frac{\gamma}{(r+1)^2}$, $\gamma\in \R$ (cf. \cite{Gla65}).
\end{remark}
\begin{remark}\label{rem:3.3}
Clearly, using the same line of reasoning as in the proof of Lemma \ref{lem:_KM}, one can extend Theorems \ref{th_III.1}, \ref{th_discretcriter}, and \ref{thContSpec} to the case of operators $\rh_{R,\gA,q}^{(l)}$ with $l\ge -1/2$.
\end{remark}

\section{Number of negative squares}\label{sec:bstates}

\subsection{The case of an arbitrary measure potential}\label{ss:iv_1M}

Let $l\ge -1/2$. To any finite non-negative Borel measure $\mu$ on $\R_+$ we associate the
quadratic form
    \be\label{0M}
\gt_{-\mu}^{(l)}[f]=\int^{\infty}_0|f'(r)|^2
dr+l(l+1)\int_0^\infty\frac{|f(r)|^2}{r^2}dr-\int^{\infty}_0|f(r)|^2 d\mu(r),
\ee
\be
\dom(\gt^{(l)}_{-\mu})=\dom(\gt^{(l)}_0).
    \ee
Here $\dom(\gt^{(l)}_0)$ denotes the form domain of the Bessel operator $\rh_0^{(l)}=-\frac{d^2}{dr^2}+\frac{l(l+1)}{r^2}$ in $L^2(\R_+)$. Note that the form $\gt_{-\mu}^{(l)}$ is closed (cf. \cite{AKM_10} and Theorem \ref{th_III.1}, Remarks \ref{rem:meas}, \ref{rem:3.3}). Denote by $\rh_{-\mu}^{(l)}=
-\frac{d^2}{dr^2}+\frac{l(l+1)}{r^2}-\mu$ the (Bessel) self--adjoint operator
associated with $\gt_{-\mu}^{(l)}$ in $L^2(\R_+)$.

   \begin{theorem}[Bargmann's bound]\label{th:bargmann}
Let $\mu$ be a finite non-negative Borel measure on $\R_+$. Then 
    \begin{equation}\label{eq:barg}
\kappa_-(\rh_{-\mu}^{(l)})< \begin{cases}\frac{1}{2l+1}\int^{\infty}_0 r d\mu(r), & l>-1/2,\\
\int^{\infty}_0 r|\log(r)| d\mu(r), & l=-1/2.
\end{cases}
    \end{equation}
\end{theorem}

\begin{remark}
If $\mu$ is absolutely continuous w.r.t. the Lebesgue measure, $\mu=q(r)dr$, the result is well known and the estimate is called the Bargmann bound and various proofs of this inequality can be found in \cite{Kato66, MMM_92, RedSim78, Sim76, Sim05}. Let us prove it for arbitrary measures.

We shall present two different proofs. The first one is elementary and is based on the classical Bargmann bound. The second one is based on the Birman--Schwinger approach and establishes a connection with the Krein string operators.
\end{remark}

\begin{proof}[The first proof.]   
Choose a nondecreasing sequence $\{q_n\}_{n=1}^\infty$ of functions $q_n:\R_+\to \R_+$ such that 
%
%
       \begin{equation}\label{4.4A}
q_{n+1}\ge q_n, \quad q_n\in L^1(\R_+; r), \quad \text{and}\quad q_n dr \xrightarrow{w} \mu.
          \end{equation}
Clearly, the operator $\rh_n:= \rh_{-q_n}^{(l)} =
-\frac{d^2}{dr^2} + \frac{l(l+1)}{r^2}- q_n$  is self-adjoint and lower semibounded.
Moreover,  it follows from \eqref{4.4A} and \eqref{0M} that
    \be\label{4.5A}
\gt_{n}^{(l)}[f] :=  \gt_{-q_n}^{(l)}[f]
\, \, \searrow \, \,   \gt_{-\mu}^{(l)}[f], \quad \text{as}\quad n\to \infty,
 \ee
for every $f\in\dom(\gt_{0}^{(l)})$, i.e., the forms $\gt_{n}^{(l)}$ approach  the form $\gt_{-\mu}^{(l)}$ from above.
Therefore, by \cite[Theorem VIII.3.11]{Kato66}, the convergence in  \eqref{4.5A} implies that the operators $\rh_{n}^{(l)}$ converge to $\rh_{-\mu}^{(l)}$ in the  strong  resolvent sense. In turn, by \cite[Theorem VIII.5.1]{Kato66},  $E_n(-\infty,0)\to E_{\rh_{-\mu}^{(l)}}(-\infty,0)$  where $E_n(\cdot)$  and  $E_{\rh_{-\mu}^{(l)}}(\cdot)$ are the spectral measures of $\rh_{n}$ and $\rh_{-\mu}^{(l)}$, respectively. Hence $\kappa_-(\rh_{-\mu}^{(l)}) = \kappa_-(\rh_{n}^{(l)})$ for $n$ large enough. Combining this relation with the classical Bargman estimates in the case $l>-1/2$ 
     \begin{equation}\label{eq:bargA}
 \kappa_-(\rh_{n}^{(l)}) \le  \frac{1}{2l+1}\int^{\infty}_0 r\, q_n(r)\, dr
 \le \frac{1}{2l+1}\int^{\infty}_0 r\, d\mu(r),   \quad   n\in\N,
    \end{equation}
we arrive at \eqref{eq:barg}  with $l>-1/2$. The case $l = -1/2$ is considered similarly.
\end{proof}

\begin{proof}[The second proof.]
      If the integral in \eqref{eq:barg} is infinite, then the claim is trivial. So, assume that it is finite.

Firstly, for $\lambda\le 0$ consider the following self-adjoint integral operator $\cK_{\mu}(\lambda)$ acting in $L^2(\R_+,d\mu)$
\be\label{eq:r0}
      (\cK_{\mu}(\lambda) f)(r)=\int_0^\infty K_l(r,s;\lambda)f(s)d\mu(s),
      \ee
where the kernel $K_l$ is the Green function of the unperturbed Bessel operator $\rh_0^{(l)}$, %
\be\label{eq:K_l01}
K_l(r,s;\lambda)=\begin{cases}
\phi_l(\lambda,r)\psi_l(\lambda,s),& r\le s\\
\phi_l(\lambda,s)\psi_l(\lambda,r),& r\ge s
                                                                       \end{cases},\quad l\ge -\frac{1}{2}.
\ee
Here
\begin{align}\label{eq:phipsi}
\phi_l(\lambda,r)=\frac{\Gamma(l+\frac{3}{2})2^{l+1}}{\sqrt{\pi}}\lambda^{-\frac{2l+1}{4}}\sqrt{\frac{\pi r}{2}}J_{l+1/2}(\sqrt{\lambda}r),\\
 \psi_l(\lambda,r)=\I \frac{\sqrt{\pi}}{\Gamma(l+\frac{3}{2})2^{l+1}}\lambda^{\frac{2l+1}{4}}\sqrt{\frac{\pi r}{2}}H^{(1)}_{l+1/2}(\sqrt{\lambda}r),
\end{align}
where $J_\nu(\cdot)$ and $H_\nu(\cdot)$ are the Bessel and the Hankel functions of order $\nu$ (see, e.g., \cite[Chapter 9]{as}).

%

Notice that  the kernel $K_l$ is positive definite if $\lambda\le 0$. Moreover, if  $\lambda_1<\lambda_2\le 0$, then $0\le \cK_\mu({\lambda_1})\le \cK_\mu({\lambda_2})$. 

Since $\mu$ is a finite measure, the identical embedding $L^2(\R_+)\to L^2(\R_+,d\mu)$ is continuous and dense. Therefore, for $\lambda<0$ the operator $\cK_{\mu}(\lambda)$ is bounded on $L^2(\R_+,d\mu)$.
Observe also that (cf. \cite[formulas (9.1.10), (9.1.3) and (9.1.11)]{as})
\[
\phi_l(0,r)=r^{l+1},\quad \psi_l(0,r)=\begin{cases}
\frac{r^{-l}}{2l+1}, & l>-\frac{1}{2}\\
\sqrt{r}|\log(r)|, & l=-\frac{1}{2},
\end{cases},\quad \quad r>0.
\]
Therefore (see \cite[Chapter III.10]{gk}), for all $\lambda\le 0$ the operator $\cK_\mu(\lambda)$ is of trace class and
\begin{align}
{\rm tr}\, \cK_\mu(\lambda)&=\int_0^\infty K_l(r,r;\lambda)d\mu(r)\le {\rm tr}\, \cK_\mu(0) \nonumber\\
&=\int_0^\infty K_l(r,r;0)d\mu(r)=\begin{cases} \frac{1}{2l+1}\int_0^\infty rd\mu(r), & l>-1/2,
\\
\int_0^\infty r|\log(r)| d\mu(r), & l=-1/2.
\end{cases}\label{eq:trace}
\end{align}

By Theorem \ref{thContSpec}, the negative spectrum of $\rh_{-\mu}^{(l)}$ consists of isolated eigenvalues. We show that $-\lambda_0<0$ is the eigenvalue of $\rh_{-\mu}^{(l)}$
if and only if  $1$ is the eigenvalue of $\cK_\mu({-\lambda_0})$.
Indeed, let $\psi_0$ be the eigenfunction of $\rh_{-\mu}^{(l)}$ corresponding to $-\lambda_0$. The latter means that the equality
\[
((\rh_0^{(l)}+\lambda_0)\psi_0, g)_{L^2}=(\psi_0,g)_{L^2_\mu} 
\]
holds true for all $g\in L^2(\R_+)$.
Therefore, we get 
\begin{align*}
(\psi_0,g)_{L^2}&=((\rh_0^{(l)}+\lambda_0)\psi_0,(\rh_0^{(l)}+\lambda_0)^{-1}g)_{L^2}
\\
&=(\psi_0,(\rh_0^{(l)}+\lambda_0)^{-1}g)_{L^2_\mu}=\int_0^\infty\left(\int_0^\infty K_l(r,s;-\lambda_0)g(s)ds\right)\psi_0(r)d\mu(r) \\
    &= \int_0^\infty \left(\int_0^\infty K_l(r,s;-\lambda_0)\psi_0(s)d\mu(s)\right)g(r)dr = (\cK_{\mu}({-\lambda_0}) \psi_0, g)_{L^2}.
\end{align*}
Since $\mu$ is a finite measure,
the identical embedding $L^2(\R_+)\to L^2(\R_+,d\mu)$ is dense.
This implies that $\psi_0$ is the eigenfunction of $\cK_\mu({-\lambda_0})$ corresponding to the eigenvalue $1$. Moreover,
 \begin{equation}\label{4.17}
\dim\ker\bigl (\rh_{-\mu}^{(l)} + \lambda_0\bigr) =  \dim\ker\bigl(I - K_{\mu}(-\lambda_0)\bigr).
  \end{equation}
%
%

Note that for any $g\in L^2(\R_+, d\mu)$ the scalar function $(\mathcal K_{\mu}(\lambda)g, g)$
strictly increases as $\lambda\uparrow 0$.
It easily follows (cf.  \cite[Theorem 4]{DM91}, the first step of the proof) that
%
%
     \begin{equation*}
\kappa_-(\rh_{-\mu}^{(l)}) = \sum_{\lambda<0}\dim\ker\bigl(I-K_{\mu}(\lambda)\bigr) \le  \kappa_-\bigl(I-K_{\mu}(0)\bigr)
       \end{equation*}
Combining this estimate  with the obvious inequality
       \begin{equation*}
\kappa_-\bigl(I-K_{\mu}(0)\bigr) = \# \{\sigma\bigl(K_{\mu}(0)\bigr)\cap(1,\infty)\} < {\rm tr}\, \cK_\mu({0}),
          \end{equation*}
%
%
and taking  formula \eqref{eq:trace} into account we complete the proof.
    \end{proof}
      \begin{remark}
(i) 
 The second proof is a modification of the Birman--Schwinger approach, see, e.g., \cite{MMM_92, RedSim78, Sim76, Sim05}.

(ii)      To the best of our knowledge (see, e.g., \cite{Kato66, RedSim78, Sim76, Sim05}) the inequality in \eqref{eq:barg} is usually  nonstrict, i.e., $\kappa_-(\rh_{-\mu})\le\int^{\infty}_0 r d\mu(r)$. However, it follows from the second proof that the inequality in \eqref{eq:barg} is indeed strict. 
      \end{remark}
\begin{corollary}
Let $q\in L^1_{loc}(\R_+)$ and $\gA$ be such that \eqref{I_brinck_q} and \eqref{I_brinck} are satisfied. Let also $l> -\frac{1}{2}$ and 
   \begin{equation}
\rh_{R,\alpha,q}^{(l)} :=
-\frac{\rD ^2}{\rD r^2}+\frac{l(l+1)}{r^2}+q(r)+\sum^{\infty}_{k=1}\alpha_{k}\delta(r-r_k).
   \end{equation}
Then 
   \begin{equation}
\kappa_-(\rh_{R,\alpha,q}^{(l)})< \frac{1}{2l+1}\Big(\int^{\infty}_0 r|q_-(r)|dr
 + \sum^{\infty}_{k=1}|\alpha^-_k|r_k\Big),
   \end{equation}
   where $q_-(x)=(q(x)-|q(x)|)/2$ and $\alpha_k^-=(\gA_k-|\gA_k|)/2$.
   \end{corollary}
   \begin{proof}
   Immediately follows from Theorem \ref{th:bargmann}.
   \end{proof}



\subsection{The case of a finite number of $\delta$-interactions}\label{ss:iv_1}

In this subsection we restrict ourselves to  the case of finitely many $\delta$-interactions,
\be\label{eq:hl_N}
\rh_{R,\gA}^{(l)}=-\frac{\rD^2}{\rD r^2}+\frac{l(l+1)}{r^2}+\sum_{k=1}^N\gA_k\delta(r-r_k),\qquad N\in\N.
\ee
We also exclude the case $l=-1/2$ in order to avoid cumbersome calculations.

\subsubsection{The boundary triplet and the corresponding Weyl function}

The operator $\rh_{R,\gA}^{(l)}$ also admits the following representation
\[
\rh_{R,\gA}^{(l)}=-\frac{\rD ^2}{\rD r^2}+\frac{l(l+1)}{r^2}+\sum_{k=1}^N\gA_k(\cdot,\delta(r-r_k))\delta(r-r_k).
\]
Note that it is a self-adjoint extension of the following symmetric operator having deficiency indices $(N,N)$
\[
\rh_{\min}^{(l)}=\rh_{R,\gA}^{(l)}\upharpoonright \dom(\rh_{\min}^{(l)}),\quad \dom(\rh_{\min})=\{f\in \dom(\rh^{(l)}_0): \ f(r_1)=...=f(r_N)=0\}. 
\]

Using the asymptotic behavior of Bessel and Hankel functions (cf. \cite[formulas (9.1.10) and (9.2.3)]{as})  
\begin{align*}
J_{l+1/2}(z)\sim \Big(\frac{z}{2}\Big)^{l+\frac{1}{2}},&\ \ z\to 0,\\
 H_{l+1/2}^{(1)}(z)\sim \sqrt{\frac{2}{\pi z}}{\rm e}^{\I(z-\frac{\pi(l+1)}{2})},&\ \ z\to \infty,\ \ (-\pi<{\rm arg}\ z<2\pi),
\end{align*}
we conclude that the defect subspace $\gN_z := \gN_z(\rh_{\min}^{(l)})$ of $\rh_{\min}^{(l)}$ is given by

\begin{equation}
\gN_z = \Span\{f_k(z,r)\}_{k=1}^N,\quad
  f_k(z,r)= \begin{cases}\phi_l(z,r)\psi_l(z,r_k),& r\le r_k\\
\phi_l(z,r_k)\psi_l(z,r),& r\ge r_k
                                                                       \end{cases},
\end{equation}
for any $z\notin\R_+$. Here the functions  $\phi_l$ and $\psi_l$ are defined by \eqref{eq:phipsi}.

First we present a boundary triplet for the operator $(\rh_{\min}^{(l)})^*$ and compute  the corresponding Weyl
function (see  Definitions \ref{def:bt} and \ref{def:wf}).
     \begin{proposition}
     \begin{itemize}
\item[(i)]   The adjoint operator $(\rh_{\min}^{(l)})^*$ is given by the differential expression
$\tau^{(l)}$ on the domain
\begin{align}
\dom\big((\rh_{\min}^{(l)})^*\big)=\{f\in &W^{2,2}(\R_+\setminus R)\cap W^{1,2}(\R_+):\\
&\, f\, \text{satisfies \eqref{eq:bc0}},\, \tau^{(l)}f\in L^2(\R_+)\}.
\end{align}
\item[(ii)]  The triplet  \ $\Pi=\{\C^N,\Gamma_0,\Gamma_1\}$, where  $\Gamma_0, \Gamma_1:\dom\big((\rh_{\min}^{(l)})^*\big)\to\C^N$ are given by
\[
 \Gamma_0f=\left(\begin{array}{c}
                            f'(r_1-)-f'(r_1+)\\
                            f'(r_2-)-f'(r_2+)\\
                            \dots\\
                            f'(r_N-)-f'(r_N+)\end{array}\right),\qquad
                                           \Gamma_1f=\left(\begin{array}{c}
                                           f(r_1)\\
                                           f(r_2)\\
                                           \dots\\
                                           f(r_N)\end{array}\right),
\]
forms  a boundary triplet for the operator $(\rh_{\min}^{(l)})^*$.
\item[(iii)] The corresponding Weyl function is given by 
\be\label{eq:33.04}
M_l(z)=\big(f_k(z,r_j)\big)_{j,k=1}^{N},\quad z\notin\R_+.
\ee
\item[(iv)]  The domain of the operator  $\rh_{R,\gA}^{(l)}$ admits the following representation 
\be\label{eq:42.04}
\dom(\rh_{R,\gA}^{(l)})=\ker(\Gamma_1+\Lambda^{-1}\Gamma_0),\quad \Lambda=\diag(\gA_1,\dots,\gA_N).
\ee
\end{itemize}
   \end{proposition}
    \begin{proof}
(i), (ii) and (iv) are straightforward. 

Let us prove (iii). By Definition \ref{def:wf}, the Weyl function is defined by 
\be\label{eq:wf}
\Gamma_1f_z=M(z)\Gamma_0f_z,\quad f_z\in\gN_z,\quad z\notin\R_+.
\ee
Setting
\be
f_z:=\sum_{k=1}^Nc_kf_k(z,r),\quad c=(c_1,...c_2)\in\C^N,
\ee
we obtain
  \begin{equation}\label{4.20}
\Gamma_1f_z=\left(\begin{array}{c}
                                           \sum_{k} c_kf_k(z,r_1)\\
                                           \sum_{k} c_kf_k(z,r_2)\\
                                           \dots\\
                                           \sum_{k} c_kf_k(z,r_N)\end{array}\right).
\end{equation}
Moreover, since
\begin{align*}
f_z'(r_j+)&-f_z'(r_j-)=\sum_{k=1}^Nc_k(f_k'(z,r_j+)-f_k'(z,r_j-))\\
&=c_j(\phi_l(z,r_j)\psi_l'(z,r_j)-\phi_l'(z,r_j)\psi_l(z,r_j))+\sum_{k\neq j}c_k\cdot 0=-c_j,
\end{align*}
we see that
 \begin{equation}\label{4.21}
\Gamma_0 f_z=
\left(\begin{array}{c}
                                           c_1\\
                                           c_2\\
                                           \dots\\
                                           c_N\end{array}\right).
 \end{equation}
Combining \eqref{4.20}  with \eqref{4.21} we arrive at  
 \eqref{eq:33.04}.
  \end{proof}
      \begin{remark}
Note that the Weyl function \eqref{eq:33.04} can be  represented in the following form
\be\label{eq:33.06}
M_l(z)=((\widetilde{\rh}^{(l)}_0-z)^{-1}\overrightarrow{\phi},\overrightarrow{\phi}), \quad \overrightarrow{\phi}=(\delta_1,\dots,\delta_N),
\ee
where $\delta_k:=\delta(\cdot-r_k)$ and $\widetilde{\rh}^{(l)}_0$ is the $[\gH_1,\gH_{-1}]$--continuation of the unperturbed Bessel operator $\rh^{(l)}_0$. Indeed, since
$(\rh^{(l)}_0-z)^{-1}=\cK$, where $\cK:=\cK_\mu^0$ is the integral operator \eqref{eq:r0}--\eqref{eq:K_l01} with $d\mu=dr$,
it is not difficult to check that $M_l$ given by \eqref{eq:33.06} coincides with \eqref{eq:33.04}.
\end{remark}

\subsubsection{The number of negative eigenvalues}
Using \eqref{eq:42.04} and \eqref{eq:33.04} and applying Proposition \ref{prop:a.1}, we arrive at the following equality
\[
\kappa_-(\rh_{R,\gA}^{(l)})=\kappa_-(-\Lambda^{-1}-M_l(0))-\kappa_-(-\Lambda^{-1}-M_l(-\infty)),
\]
where
\[
M_l(0):=\lim_{\lambda\uparrow-0}M_l(\lambda),\quad M_l(-\infty):=\lim_{\lambda\downarrow-\infty}M_l(\lambda).
\]
Using \cite[formulas (9.2.1) and (9.2.3)]{as}, it is not difficult to see that $M_l(-\infty)=\bold{0}$. Moreover, using \cite[formulas (9.1.7) and (9.1.9)]{as}, we obtain
\be\label{eq:33.05}
M_l(0):=\frac{1}{2l+1}\left(\begin{array}{ccccc} r_1 & r_1^{l+1}r_2^{-l} & r_1^{l+1}r_3^{-l} & \dots  &r_1^{l+1}r_N^{-l}\\
r_1^{l+1}r_2^{-l} & r_2 & r_2^{l+1}r_3^{-l} & \dots &r_2^{l+1}r_N^{-l}\\
r_1^{l+1}r_3^{-l} & r_2^{l+1}r_3^{-l} & r_3 & \dots &r_3^{l+1}r_N^{-l}\\
 \dots& \dots& \dots& \dots & \dots \\
r_1^{l+1}r_N^{-l} & r_2^{l+1}r_N^{-l} & r_3^{l+1}r_N^{-l} & \dots & r_N\\
\end{array}\right).
\ee
Finally, noting that $\kappa_-(-\Lambda^{-1})=\kappa_+(\Lambda)=\kappa_+(\gA)$, we prove the following result.
%
%
   \begin{theorem}\label{th:kappa-N}
Let $\rh_{R,\gA}^{(l)}$ be the operator given by \eqref{eq:hl_N}. Then
\be\label{eq:32.01}
\kappa_-(\rh_{R,\gA}^{(l)})=\kappa_+(M_{R,\gA}^l)-\kappa_+(\gA),
\ee
where
\begin{align}
 M_{R,\gA}^l&=(2l+1)(\Lambda^{-1}+M_l(0))\nonumber\\
 &=\left(\begin{array}{ccccc} \frac{2l+1}{\gA_1}+r_1 & r_1^{l+1}r_2^{-l} & r_1^{l+1}r_3^{-l} & \dots  & r_1^{l+1}r_N^{-l}\\
r_1^{l+1}r_2^{-l} & \frac{2l+1}{\gA_2}+r_2 & r_2^{l+1}r_3^{-l} & \dots &r_2^{l+1}r_N^{-l}\\
r_1^{l+1}r_3^{-l} & r_2^{l+1}r_3^{-l} & \frac{2l+1}{\gA_3}+r_3 & \dots &r_3^{l+1}r_N^{-l}\\
 \dots& \dots& \dots& \dots & \dots \\
r_1^{l+1}r_N^{-l} & r_2^{l+1}r_N^{-l} & r_3^{l+1}r_N^{-l} & \dots & \frac{2l+1}{\gA_N}+r_N\\
\end{array}\right).\label{eq:32.02}
\end{align}
\end{theorem}

As an immediate corollary we get the following statement.

\begin{corollary}\label{cor:4.6}
Let $\rh_{R,\gA}^{(l)}$ be the operator given by \eqref{eq:hl_N}. Then
\be\label{eq:32.01B}
\kappa_-(\rh_{R,\gA}^{(l)})\le \kappa_-(\gA).
\ee
\end{corollary}

\begin{proof}
It suffices to note that $\kappa_+(M_{R,\gA}^l)\le N$ and $\kappa_+(\gA)=N-\kappa_-(\gA)$. Therefore, using \eqref{eq:32.01}, we arrive at \eqref{eq:32.01B}.
\end{proof}

\begin{corollary}\label{cor:4.7}
Let $\kappa_-(\gA)=n\le N$. Denote the negative intensities by $\gA^-=\{\gA_k^-\}_{k=1}^n$ and the corresponding  centers by $R^-=\{r_k^-\}_{k=1}^n$. Then
\be
 \kappa_-(\rh_{R,\gA}^{(l)})\le\kappa_+(M_{R^-,\gA^-}^l),
 \ee
 where
 \be
M_{R^-,\gA^-}^l=\left(\begin{array}{cccc} \frac{2l+1}{\gA_1^-}+r_1^- & (r_1^-)^{l+1}(r_2^-)^{-l} & \dots  & (r_1^-)^{l+1}(r_n^-)^{-l}\\
(r_1^-)^{l+1}(r_2^-)^{-l} & \frac{2l+1}{\gA_2^-}+r_2^- &  \dots & (r_2^-)^{l+1}(r_n^-)^{-l}\\
  \dots& \dots& \dots & \dots \\
(r_1^-)^{l+1}(r_n^-)^{-l} & (r_2^-)^{l+1}(r_n^-)^{-l} &  \dots & \frac{2l+1}{\gA_n^-}+r_n^-\\
\end{array}\right).\label{eq:33.03}
\ee
\end{corollary}

\begin{proof}
To prove the claim it suffices to note that
\[
\kappa_-(\rh_{R,\gA}^{(l)})\le \kappa_-(\rh_{R^-,\gA^-}^{(l)}),
\]
and then to apply Theorem \ref{th:kappa-N} to the operator $\rh_{R^-,\gA^-}^{(l)}$.
\end{proof}


\begin{example}\label{ex:N=1}
Assume that $N=1$, i.e., consider the Hamiltonian
\be
\rh_1^{(l)}:=-\frac{d^2}{dr^2}+\frac{l(l+1)}{r^2}+\gA_1\delta(r-r_1),
\ee
where $\gA_1<0$ and $r_1>0$.
By Theorem \ref{th:kappa-N}, we obtain
\be
\kappa_-(\rh_1^{(l)})=\begin{cases}
0, & |\gA_1|r_1\le 2l+1,\\
1, & |\gA_1|r_1>2l+1.
\end{cases}
\ee
\end{example}

\begin{example}\label{ex:N=2}
Assume now that $N=2$, i.e., consider the Hamiltonian
\be
\rh_2^{(l)}:=-\frac{d^2}{dr^2}+\frac{l(l+1)}{r^2}+\gA_1\delta(r-r_1)+\gA_2\delta(r-r_2),
\ee
where $\gA_j\in \R\setminus\{0\}$ and $r_2>r_1>0$.

If $\gA=\{\gA_1,\gA_2\}\subset(0,+\infty)$, then, by Corollary \ref{cor:4.6}, $\kappa_-(\rh_2^{(l)})=0$.

Assume that either $\gA_1<0$ or $\gA_2<0$, that is $\kappa_-(\gA)=1$. Then, by Theorem \ref{th:kappa-N},
\be
\kappa_-(\rh_2^{(l)})=\begin{cases}
0, & \det\, M_2^l\le 0,\\
1, & \det\, M_2^l> 0.
\end{cases}
\ee
Here
\be\label{eq:m2}
M_2^l=\begin{pmatrix}
\frac{2l+1}{\gA_1}+r_1 & r_1^{l+1}r_2^{-l} \\
 r_1^{l+1}r_2^{-l} & \frac{2l+1}{\gA_2}+r_2
\end{pmatrix}.
\ee

If both $\gA_1$ and $\gA_2$ are negative, then
\be
\kappa_-(\rh_2^{(l)})=\kappa_+(M_2^l).
\ee
In particular, $\kappa_-(\rh_2^{(l)})=0$ precisely if
\be\label{eq:4.38}
|\gA_j|r_j<2l+1,\quad (j=1,2),\quad \text{and}\quad \left(\frac{2l+1}{\gA_1}+r_1\right) \left(\frac{2l+1}{\gA_2}+r_2\right)>\frac{r_1^{2l+2}}{r_2^{2l}}.
\ee
Moreover,
$\kappa_-(\rh_2^{(l)})=2$ if and only if
\be\label{eq:4.39}
|\gA_j|r_j\ge 2l+1,\quad (j=1,2),\quad \text{and}\quad \left(\frac{2l+1}{\gA_1}+r_1 \right) \left(\frac{2l+1}{\gA_2}+r_2 \right)<\frac{r_1^{2l+2}}{r_2^{2l}}.
\ee
If $\gA_1,\gA_2$ and $r_1,r_2$ do not satisfy neither \eqref{eq:4.38} nor \eqref{eq:4.39}, then $\kappa_-(\rh_2^l)=1$.
\end{example}

\begin{example}\label{ex:N>2}
Assume now that $N\ge 3$. Assume also that $\gA_k>0$ for all $k\neq 2$ and $\gA_2<0$. Therefore, $\kappa_+(\gA)=N-1$. By  Corollary \ref{cor:4.6}, $\kappa_+(M_{R,\gA}^l)\ge N-1$. In particular, the operator $\rh_{R,\gA}$ is positive if and only if $\kappa_+(M_{R,\gA}^l)= N-1$. Furthermore, if $\kappa_-(M_2^l)=1$, where $M_2^l$ is given by \eqref{eq:m2}, then the operator $\rh_{R,\gA}$ is positive. However, this condition is only sufficient. 
\end{example}


\subsubsection{Some necessary and sufficient conditions}
We begin with the following simple but useful result. 

\begin{corollary}\label{cor:4.10}
(i) For the Hamiltonian $\rh_{R,\gA}^{(l)}$ to have a maximal number of
negative squares, i.e., for the equality $\kappa_-(\rh_{R,\gA}^{(l)})= N$  to hold,
it is necessary  that  $\gA=\gA^-$ and
  \begin{equation}\label{4.40}
 |\alpha_k|r_k> 2l+1, \qquad   k\in\{1,\ldots, N\}.
      \end{equation}
In particular, $\kappa_-(\rh_{R,\gA}^{(l)})\le N-1$ if at least of one of the inequalities in \eqref{4.40} is not satisfied.

(ii) If $\gA=\gA^-$ and the Hamiltonian $\rh_{R,\gA}^{(l)}$ is positive, i.e.,
$\kappa_-(\rh_{R,\gA}^{(l)}) = 0$, then 
   \begin{equation}\label{4.41}
|\gA_k|r_k \le {2l+1},  \qquad k\in\{1,\ldots, N\}.
    \end{equation}
In particular, $\kappa_-(\rh_{R,\gA}^{(l)})\ge 1$ if $\gA=\gA^-$ and at least of one of the inequalities in \eqref{4.41} is not satisfied.
  \end{corollary}
\begin{proof}
(i)  By Theorem \ref{th:kappa-N},   $\kappa_-(\rh_{R,\gA}^{(l)})= N$ if and only if $\kappa_+(M_{R,\gA}^l) = N$ and $\kappa_+(\alpha)=0$, that is, all $\gA_k\le 0$ and the matrix $M_{R,\gA}^l$ is positive,  $M_{R,\gA}^l > 0$. Therefore, $\frac{2l+1}{\alpha_k}+r_k=(M^l_{R,\gA} e_k, e_k)> 0$, where $\{e_k\}_{k=1}^N$ is the standard orthonormal basis of $\C^N$.

(ii) Since $\alpha=\alpha_-$, Theorem \ref{th:kappa-N}  yields the equality $\kappa_+(M_{R,\gA}^l)=0$, i.e., the matrix $M_{R,\gA}^l$ is nonpositive, $M_{R,\gA}^l\le 0$. Hence all its diagonal entries are nonpositive and hence we arrive at \eqref{4.41}.
  \end{proof}

\begin{remark}
Notice that \eqref{4.41} is no longer necessary for the positivity of $\rh_{R,\gA}^{(l)}$ if $\gA\neq \gA^-$. Namely, let $l=0$  and set $N=2$, $\gA_1=-2$, $\gA_2=2$ and $r_1=1$, $r_2=3/2$. Hence $\kappa_+(\gA)=1$ and
\[
M_{R,\gA}^0=\begin{pmatrix}
\frac{1}{2} & 1 \\
 1 & 2
\end{pmatrix}.
\]
Clearly, $\det M_{R,\gA}^0=0$ and hence (see  Example \ref{ex:N=2}) we obtain $\kappa_-(\rh_{R,\gA}^{(l)})=0$. However,
\[
|\gA_1|r_1=2\cdot 1=2>1,
\]
and hence the first inequality in \eqref{4.41} is not satisfied.
\end{remark}

Our next aim is to apply Gershgorin's Theorem (see, e.g., \cite[Theorem 10.6.1]{LanTis85}) to the matrix $M^l_{R,\gA}$ in order to estimate  $\kappa_-(\rh_{R,\gA}^{(l)})$.

\begin{theorem}[Gershgorin] \label{th:gersh}
All eigenvalues of a matrix $A=(a_{ij})_{i,j=1}^n\in \C^{n\times n}$
are contained in the union of Gershgorin's disks
\be\label{eq:gersh}
\mathbb{D}_i=\Big\{z\in\C:\, |z-a_{ii}|\le \sum_{j\neq i}|a_{ij}|\Big\},\quad i\in\{1,...,n\}.
\ee

Furthermore, a set of $m$ discs having no point in common with the remaining $n-m$ discs contains precisely $m$ eigenvalues of $A$.
\end{theorem}

\begin{proposition}\label{propGershgorin}
Let $\rh_{R,\gA}^{(l)}$ be given by \eqref{eq:hl_N} and $\alpha=\alpha^-$. Assume also that $\{1,...,N\}=\Omega_+\cup\Omega_-$ and there exists a sequence $\{b_k\}^N_1\subset\R_+$ such that
     \begin{align}\label{4.54}
\frac{2l+1}{|\alpha_k|} <r_k \Big[1-\sum^{k-1}_{j=1} \frac{b_j}{b_k} \Big(\frac{r_j}{r_k}\Big)^{l+1} - \sum^N_{j=k+1}\frac{b_j}{b_k} \Big(\frac{r_k}{r_j}\Big)^l \Big], \quad (k\in\Omega_+),
%
%
     \end{align}
and
    \begin{equation}\label{4.56}
\frac{2l+1}{|\alpha_k|} \ge r_k \Big[1+\sum^{k-1}_{j=1} \frac{b_j}{b_k} \Big(\frac{r_j}{r_k}\Big)^{l+1} + \sum^N_{j=k+1}\frac{b_j}{b_k} \Big(\frac{r_k}{r_j}\Big)^l \Big], \quad (k\in\Omega_-).
    \end{equation}
%
%
Then $\kappa_-(\rh_{R,\gA}^{(l)}) = \#( \Omega_+)$.
      \end{proposition}
  \begin{proof}
Since $\alpha = \alpha^-,$  $\kappa_+(\alpha)=0$. Hence, by Theorem \ref{th:kappa-N}, $\kappa_-(\rh_{R,\gA}^{(l)}) = \kappa_+(M_{R, \gA}^l)$.
Note that $\kappa_-(B^* M^l_{R,\alpha}B)=\kappa_-(M^l_{R,\alpha})$ for any non-singular matrix $B\in\Bbb C^{N\times N}$.  

Set $B= \diag(b_1,b_2, \ldots, b_N)$ with $b_j>0, j\in \{1,\ldots, N\}$,  and apply Theorem \ref{th:gersh}
to the matrix $B^* M^l_{R,\alpha}B$ in order to find $\kappa_+(M^{(l)}_{R,\alpha})$.
Namely, inequalities \eqref{4.54} imply that the Gershgorin disc $\mathbb{D}_k$ satisfies $\mathbb{D}_k\cap(-\infty,0]=\emptyset$ if  $k\in\Omega_+$.
%
%
%
%
On the other hand, inequalities  \eqref{4.56} mean that $\mathbb{D}_k\cap (0,+\infty)=\emptyset$ for all $k\in\Omega_-$. Thus, two sets $\cup_{k\in\Omega_+}\mathbb{D}_k$ and $\cup_{k\in\Omega_-}\mathbb{D}_k$ have no common points and hence, by Theorem \ref{th:gersh},  $\kappa_+(B^*M_{R, \gA}^lB) = \#(\Omega_+)$. Therefore, $\kappa_+(M_{R, \gA}^l) =  \#(\Omega_+)$ and we are done.
  \end{proof}

\begin{remark}
Note that  the possibility to choose a sequence  $\{b_k\}_1^N$ allows to apply the Gershgorin Theorem to the matrix $M^{l}_{R,\gA}$. Namely, the inequalities
\eqref{4.56} might be satisfied with $b_1=\ldots = b_N=1$ for some $r_k$ and $\alpha_k$. However, this is not the case for the  inequalities \eqref{4.54}. For instance, they are not satisfied for any $k<N$ if
$b_1=\ldots = b_N=1$ and $l= 0$. Indeed, in this case \eqref{4.54} is equivalent to the following inequality %
\[
\frac{1}{|\gA_k|}< r_k- \sum_{j=1}^{k-1}r_j -r_k(N-k).
\]
Clearly, if $N\ge 2$, then  the righthand side of the last inequality is positive precisely if $k=N$ and $r_N>\sum_{k=1}^{N-1}r_k$.
   \end{remark}

  \begin{proposition}
Let $\rh_{R,\gA}^{(l)}$ be given by \eqref{eq:hl_N} and $\alpha=\alpha^-$. Assume that there exists $\varepsilon\in (0,1)$ such that
\begin{equation}\label{4.60}
\Big(\frac{r_1}{r_2}\Big)^{l+1}<\frac{\varepsilon^2(1-\varepsilon)}{2}, \quad \Big(\frac{r_1}{r_3}\Big)^{l+1} < \frac{\varepsilon^3}{6N}, \quad  \Big(\frac{r_2}{r_3}\Big)^{l+1}
< \frac{\varepsilon^2}{6N}
  \end{equation}
and
\be\label{4.60B}
\Big(\frac{r_k}{r_{k+1}}\Big)^l\le \frac{\varepsilon}{3N},\quad k\in\{3,\ldots,N\}.
\ee
%
%

Then $\kappa_-(\rh_{R,\gA}^{(l)}) = 2$ if

    \begin{equation}\label{4.58}
\frac{2l+1}{|\alpha_k|}<r_k(1-\varepsilon)<r_k,\quad k\in\{1,2\}.
    \end{equation}
and
    \begin{equation}\label{4.59}
\frac{2l+1}{|\alpha_k|}\ge r_k(1+\varepsilon)>r_k, \quad k\in\{3, \ldots, N\}.
    \end{equation}
     \end{proposition}
   \begin{proof}
To prove the statement it suffices to show that there is a sequence of positive numbers $\{b_k\}_1^N$ such that the conditions of Proposition   \ref{propGershgorin}  are satisfied with $\Omega_+=\{1,2\}$ and $\Omega_-=\{3,...,N\}$.

Set 
   \begin{equation}\label{4.61}   
b_1:=1,\quad b_2:=\frac{\varepsilon(2-\varepsilon)}{2}, \qquad b_k:=\frac{\varepsilon^2}{2(N-2)}, \qquad k\in\{3,\ldots,N\}.
   \end{equation}
Let us check the inequalities  \eqref{4.54} with $k\in \Omega_+=\{1,2\}$. Firstly, using  \eqref{4.61}, we easily get
    \[
\sum_{j=2}^N\frac{b_j}{b_1}\left(\frac{r_1}{r_j}\right)^l < \sum_{j=2}^N b_j = b_2 + \sum_{j=3}^N b_j =
\frac{\varepsilon(2-\varepsilon)}{2} + \frac{\varepsilon^2}{2}=\varepsilon,
   \]
and hence \eqref{4.58} implies \eqref{4.54} in the case $k=1$.

Next, let $k=2$.   Using the first inequality in \eqref{4.60} and noting that $r_2<r_k$ for  $k>2$, we get
    \begin{align*}
    \frac{b_1}{b_2} \Big(\frac{r_1}{r_2}\Big)^{l+1} &+ \sum^N_{j=3}\frac{b_j}{b_2} \Big(\frac{r_2}{r_j}\Big)^l < \frac{1}{b_2} \Big(\frac{r_1}{r_2}\Big)^{l+1} + \sum^N_{j=3}\frac{b_j}{b_2}\\
&=
 \frac{\varepsilon^2(1-\varepsilon)}{\varepsilon(2-\varepsilon)} + \frac{\varepsilon^2}{\varepsilon(2-\varepsilon)}=
  \frac{\varepsilon(1-\varepsilon)+\varepsilon}{2-\varepsilon} =\varepsilon. 
    \end{align*}
Combining this inequality with \eqref{4.58}, we arrive at   \eqref{4.54} with $k=2$.

Finally, we check the conditions \eqref{4.56}.  Starting with inequalities \eqref{4.59}
and using inequalities \eqref{4.60} we obtain for $k\in\{3,\ldots,N\},$
   \begin{align*}
  & \sum^{k-1}_{j=1} \frac{b_j}{b_k} \Big(\frac{r_j}{r_k}\Big)^{l+1} + \sum^N_{j=k+1}\frac{b_j}{b_k} \Big(\frac{r_k}{r_j}\Big)^l \\
   &=\frac{b_1}{b_k}\Big(\frac{r_1}{r_k}\Big)^{l+1}+\frac{b_2}{b_k}\Big(\frac{r_2}{r_k}\Big)^{l+1}
   +\sum_{j=3}^{k-1}\Big(\frac{r_j}{r_k}\Big)^{l+1}+\sum_{j=k+1}^{N}\Big(\frac{r_k}{r_j}\Big)^{l}\\
   & = \frac{\varepsilon(N-2)}{3N} + \frac{\varepsilon(1-\varepsilon/2)(N-2)}{3N} + \frac{\varepsilon(N-3)}{3N} \le \varepsilon.
   \end{align*}

Thus, inequalities \eqref{4.56} are verified for $k\in \Omega_-=\{3,...,N\}$. It remains to apply Proposition \ref{propGershgorin}.
   \end{proof}

\subsubsection{Bargmann's bound}\label{sss:iv.2.1}
Assume that $\Lambda=\Lambda^-$, i.e., all intensities $\alpha_k$ are negative, $\gA_k=\gA_k^-$. Note that
\be\label{eq:32.03}
\kappa_+((\Lambda^-)^{-1}+M_{l}(0))=\kappa_-((2l+1)I_N-(2l+1)|\Lambda^-|^{1/2}M_l(0)|\Lambda^-|^{1/2}).
\ee
Moreover, it is easy to see that
\begin{align}
M_R^\Lambda&:=(2l+1)|\Lambda^-|^{1/2}M_l(0)|\Lambda^-|^{1/2}\nn \\
&=\left(\begin{array}{cccc} r_1^-|\gA_1^-| & \frac{(r_1^-)^{l+1}}{(r_2^-)^{l}}\sqrt{\gA_1^-\gA_2^-} &  \dots  & \frac{(r_1^-)^{l+1}}{(r_n^-)^{l}}\sqrt{\gA_1^-\gA_N^-}\\
\frac{(r_1^-)^{l+1}}{(r_2^-)^{l}}\sqrt{\gA_1^-\gA_2^-} & r_2^-|\gA_2^-| &  \dots & \frac{(r_2^-)^{l+1}}{(r_n^-)^{l}}\sqrt{\gA_2^-\gA_N^-}\\
 \dots& \dots& \dots & \dots \\
\frac{(r_1^-)^{l+1}}{(r_N^-)^{l}}\sqrt{\gA_1^-\gA_N^-} & \frac{(r_2^-)^{l+1}}{(r_N^-)^{l}}\sqrt{\gA_2^-\gA_N^-}  & \dots & r_N^-|\gA_N^-|\\
\end{array}\right).\label{eq:4.42}
\end{align}

The latter enables us to present one more of the estimate \eqref{eq:barg} in the special case $\mu(x)=\sum_{k=1}^N\alpha_k\delta(x-x_k)$.

\begin{corollary}\label{cor:bb}
Let $\rh_{R,\gA}^{(l)}$ be given by \eqref{eq:hl_N}. Let also $l>-\frac{1}{2}$ and $\alpha=\alpha^-$. Then:
\be\label{eq:bargmann}
\kappa_-(\rh_{R,\gA}^{(l)})< \frac{1}{2l+1}\sum_{k=1}^N|\gA_k|r_k\quad \text{ (Bargmann's bound)}
\ee
In particular, $\kappa_-(\rh_{R,\gA}^{(l)})=0$ if
\be\label{eq:bargmann0}
\sum_{k=1}^N|\gA_k|r_k\le 2l+1.
\ee
\end{corollary}

\begin{proof}
 Combining \eqref{eq:32.01} with \eqref{eq:32.03} and noting that $\Lambda=\Lambda^-$, we get
\begin{align}
&\kappa_-(\rh_{R,\gA}^{l})\le \kappa_-((2l+1)I_N-M_R^\Lambda)=\sum_{\lambda_j(M_R^\Lambda)>2l+1}1\nn \\ &< \sum_{\lambda_j(M_R^\Lambda)>2l+1}\frac{\lambda_j(M_R^\Lambda)}{2l+1}\le \frac{\mathrm{tr} \big(M_R^\Lambda\big)}{2l+1}=\frac{1}{2l+1}\sum_{k=1}^n|\gA_k^-|r_k^-.\nn
\end{align}
\end{proof}

\begin{remark}
(i) Example \ref{ex:N=1} shows that in the case of a one-center $\delta$-interaction the Bargmann bound \eqref{eq:bargmann} provides a criterion since in this case the inequality $|\gA_1|r_1\le 2l+1$ is not only sufficient for the equality $\kappa_-(\rh_{R,\gA}^{(l)})=0$ but is also necessary.

(ii) It is interesting to compare inequalities \eqref{4.40} with the necessary condition  implied by
the Bargmann estimate  \eqref{eq:bargmann}. Clearly, if $\kappa_-(\rh_{R,\gA}^{(l)})=N$, then \eqref{eq:bargmann}  yields only the following estimate
    \begin{equation*}
\sum^N_{k=1}|\alpha_k|r_k>N(2l+1)
    \end{equation*}
in place of stronger inequalities  \eqref{4.40}.
\end{remark}

\subsection{Absence of bound states}\label{sss:iv.2.2}

Consider the operator $\rh_{-\mu}^{(l)}$ defined by \eqref{0M}, where $\mu$ is a positive measure on $\R_+$. Assume for simplicity that $l>-\frac{1}{2}$. The next fact immediately follows from the Bargmann estimate \eqref{eq:bargmann} and in the case of absolutely continuous  measures $\mu(r)=q(r)dr$ was first established by R. Jost and A. Pais (cf. \cite{Sim76, RedSim78}).

\begin{lemma}\label{cor:4.5}
 The operator $\rh_{-\mu}^{(l)}$ is nonnegative if
\begin{equation}\label{eq:iv.19}
\frac{1}{2l+1} \int_{\R_+}rd\mu(r)\le 1.
\end{equation}
\end{lemma}

\begin{proof}
The proof is immediate from Theorem \ref{th:bargmann}.
\end{proof}

\begin{lemma}\label{cor:4.6}
Let $l=0$ and $\mu$ be a finite positive measure on $\R_+$. Then for the equality $\kappa_-(\rh_{-\mu}^{(0)})=0$ it is necessary that
\be\label{eq:33.10}
\sup_{r\in\R_+}r \int_{(r,+\infty)}d\mu(t)\le 1,
\ee
and sufficient that
\be\label{eq:33.10A}
\sup_{r\in\R_+}r \int_{(r,+\infty)}d\mu(t)\le \frac{1}{4}.
\ee
\end{lemma}

\begin{proof}
Firstly, by \eqref{0M}, we observe that $\kappa_-(\rh_{-\mu}^{(0)})=0$ precisely if
\[
\|u'\|^2_{L^2(\R_+)}\ge \|u\|^2_{L^2(\R_+,d\mu)},\quad u\in \dom(\rh_{-\mu}^{(0)})=W^{1,2}(\R_+).
\]
The latter holds true precisely if the operator $\cK_\mu(0)$ given by \eqref{eq:r0}--\eqref{eq:K_l01} with $l=0$ satisfies
\be\label{eq:4.32}
\|\cK_{\mu}(0)\|_{L^2_{\mu}}\le 1.
\ee
Therefore, using the Kac--Krein criteria \cite[Theorems 1, 3]{KK58}, we arrive at the following implications
\[
\sup_{r>0} r (V(+\infty)-V(r))\le \frac{1}{4} \quad \Rightarrow\quad \eqref{eq:4.32}\quad \Rightarrow \quad \sup_{r>0} r (V(+\infty)-V(r))\le 1.
\]
Here $V(r)=\mathrm{Var}_{[0,r]}\mu=\int_{[0,r]}d\mu(r)$. Since \eqref{eq:4.32} is equivalent to the equality $\kappa_-(\rh_{-\mu}^{(0)})=0$, we are done.
%
\end{proof}

\begin{remark}
Let us mention that the sufficiency of a stronger version of \eqref{eq:33.10A} for the equality $\kappa_-(\rh_{-\mu}^{(0)})=0$ is well known in the case of continuous potentials. Namely (see, e.g.,  \cite{RedSim78} and also \cite[Theorem 2]{ek_87}), the inequality
\[
q_-(r)\le \frac{1}{4r^2} \quad \text{a.e. on }\quad \R_+
\]
implies $\kappa_-(\rh_q^{(0)})=0$.

It is also interesting to compare necessary condition \eqref{eq:33.10} with the necessary condition obtained in \cite[Theorem 2]{ek_87}. Namely, Theorem 2 from \cite{ek_87} states that for any function $\tilde{q}:\R_+\to \R_+$ such that $\tilde{q}\ge \frac{1}{4r^2}$ on $\R_+$ and $\tilde{q}(r)\ge \frac{1}{4r^2}+c$, $c>0$, on some interval of $\R_+$, there is $q$ such that $\frac{1}{4r^2}\le q(r)\le \tilde{q}(r)$ and $\kappa_-(\rh_{-q}^{(0)})\ge 1$.
\end{remark}

Let us now restrict our considerations to the case of a finite number of $\delta$-interaction, that is, let us consider the operator $\rh_{R,\gA}^{(l)}$ given by \eqref{eq:hl_N}. 
The next fact immediately follows from \eqref{eq:32.03}. 
\begin{lemma}\label{lem:4.12}
 Let the operator $\rh_{R,\gA}^{(l)}$ be given by \eqref{eq:hl_N} with  $\gA_k<0$ for all $k\in\{1,...,N\}$. Then the operator $\rh_{R,\gA}^{(l)}$ is nonnegative if and only if $\|M_{R}^\Lambda\|\le 2l+1$, where the matrix $M_R^\Lambda$ is given by \eqref{eq:4.42}.
\end{lemma}

\begin{proof}
Combining \eqref{eq:32.01} with \eqref{eq:32.03} and noting that $\kappa_+(\gA)=0$, we get
\[
\kappa_-(\rh_{R,\gA}^{(l)})=\kappa_-\big((2l+1)I_N-M_R^\Lambda\big).
\]
Therefore, $\kappa_-(\rh_{R,\gA}^{(l)})=0$ precisely if $\|M_R^\Lambda\|\le 2l+1$.
\end{proof}

\begin{corollary}
Let the assumptions of Lemma \ref{lem:4.12} be satisfied. Then $\kappa_-(\rh_{R,\gA}^{(l)})=0$ if and only if
\be\label{eq:4.34}
\left(\begin{array}{ccccc} r_1 & r_1^{l+1}r_2^{-l} &  \dots  & r_1^{l+1}r_N^{-l}\\
r_1^{l+1}r_2^{-l} & r_2 & \dots &r_2^{l+1}r_N^{-l}\\
 \dots& \dots&  \dots & \dots \\
r_1^{l+1}r_N^{-l} & r_2^{l+1}r_N^{-l} & \dots & r_N\\
\end{array}\right)\le
\left(\begin{array}{cccc}
\frac{2l+1}{|\alpha_1|} & 0 & \dots & 0\\
0 &  \frac{2l+1}{|\alpha_2|}& \dots & 0\\
\dots & \dots & \dots &\dots\\
0 & 0 & \dots &\frac{2l+1}{|\alpha_N|}\\
\end{array}\right). 
\ee
\end{corollary}

\begin{proof}
Combining \eqref{eq:4.42} with \eqref{eq:33.05} and using Lemma \ref{lem:4.12}, we prove the claim.
\end{proof}


%
%

%
%

\begin{corollary}
Let the assumptions of Lemma \ref{lem:4.12} be satisfied.  If
\be\label{eq:4.50}
\frac{1}{|\gA_k|}\ge \frac{1}{2l+1}\Big(\sum_{j=1}^{k-1}\frac{r_j^{l+1}}{r_k^{l}}+r_k\sum_{j=k}^N\Big(\frac{r_k}{r_j}\Big)^l\Big)
\ee
for every  $k\in\{1,...,N\}$, then $\kappa_-(\rh_{R,\gA}^{(l)})=0$, i.e., the operator $\rh_{R,\gA}^{(l)}$ is positive.
   \end{corollary}
\begin{proof} The proof is immediate from Proposition  \ref{propGershgorin}
with  $\Omega_+=\emptyset$, $\Omega_- = \{1,...,N\}$ and $b_1=\ldots =b_N=1$.
  \end{proof}

\begin{remark}
In general, inequalities \eqref{eq:4.50} do not imply the Jost--Pais estimate \eqref{eq:bargmann0} and vise versa. Namely,
let $N=2$. Assume that $\gA_1$ and $r_1$ are such that $|\gA_1|r_1=l+\frac{1}{2}$. Then \eqref{eq:bargmann0} holds true if
\[
r_2|\gA_2|\le l+\frac{1}{2}.
\]
However, Gershgorin's estimates \eqref{eq:4.50} hold true if
\[
r_2|\gA_2|\le (2l+1)\frac{r_2^{l+1}}{r_1^{l+1}+r_2^{l+1}},
\]
which is weaker than the above estimate since $r_1<r_2$.

On the other hand, set $l=0$, $\gA_1=-2$, $\gA_2=-\frac{1}{3}$ and $r_1=\frac{1}{3}$, $r_2=1$. Then $|\gA_1|r_1=\frac{2}{3}\ge \frac{1}{2}$ and hence estimates \eqref{eq:4.50} are not satisfied. However, 
\[
\sum_{k=1}^2|\gA_k|r_k=2\cdot\frac{1}{3}+\frac{1}{3}\cdot1=1,
\]
and hence \eqref{eq:bargmann0} holds true.
\end{remark}

\section{Schr\"odinger operators with $\delta$-shells}\label{sec:v}

The main aim of this section is to extend the results on spectral properties of Schr\"odinger operators with point interactions from the case of one dimension to the multidimensional case.

\subsection{Self-adjointness}

We begin with the following result.
\begin{theorem}\label{th:sa}
Let the operators $\rH_{R,\gA}$ and $B_{R,\gA}$ be given by \eqref{I_01} and \eqref{I_07}, respectively. Then $\rH_{R,\gA}$  is self-adjoint if and only if the matrix $B_{R,\gA}$ is also self-adjoint. In particular, if $n_\pm(B_{R,\gA})=1$, then $n_\pm(\rH_{R,\gA})=\infty$.
\end{theorem}
\begin{proof}
By Theorem \ref{th_KM}$(i)$, the operators $\rh_{R,\gA}^{(l)}$, $l\in\N_0$, are self-adjoint if and only if so is $B_{R,\gA}$ and, moreover, $n_\pm(\rh_{R,\gA}^{(l)})=n_\pm(B_{R,\gA})$. Therefore, representation \eqref{eq:31.01}--\eqref{eq:31.03} completes the proof.
\end{proof}
In \cite{KM_09, AKM_10}, several simple necessary and sufficient self-adjointness conditions have been obtained.
\begin{corollary}\label{cor:51:01}
The Hamiltonian $\rH_{R,\gA}$ is self-adjoint for any
$\gA=\{\alpha_k\}_{k=1}^\infty\subset\R$ whenever
       \begin{equation}\label{eq:51.01}
\sum_{k=1}^\infty \gd_k^2=\infty.
\end{equation}
\end{corollary}
\begin{proof}
The proof is immediate from Theorem \ref{th:sa} and \cite[Proposition 5.7]{KM_09}.
\end{proof}
It was observed in \cite{KM_09} that condition \eqref{eq:51.01} is sharp for the self-adjointness of $\rh_{R,\gA}^{(0)}$. Therefore, combining \cite[Proposition 5.9]{KM_09} with Theorem \ref{th:sa}, we immediately arrive at the following result. 
\begin{corollary}\label{cor:51.02}
Let \ $\sum_{k\in\N}\gd_k^2<\infty$ and
$\gd_{k-1}\gd_{k+1}\geq \gd_k^2$ for all $k\in\N$. If
\begin{equation}\label{I_06}
\sum_{k=1}^\infty \gd_{k+1}\bigl|\alpha_k+\frac{1}{\gd_{k}}+\frac{1}{\gd_{k+1}}\bigr|<\infty,
\end{equation}
then the operator $\rH_{R,\gA}$ is symmetric with $n_\pm(\rH_{R,\gA})=\infty$.
\end{corollary}
In the case  $\{d_k\} \in l_2$, the question on self-adjointness of the operator $\rH_{R,\gA}$ is quite subtle. Several further necessary conditions can be found in \cite[\S 5.2]{KM_09} (see also \cite{Bra85, AKM_10}). Let us demonstrate this by the following example.
\begin{corollary}\label{cor:51.03}
Let $\gA=\{\gA_k\}_{k=1}^\infty\subset \R$ and let $R=\{r_k\}_{k=1}^\infty$ be given by $r_k-r_{k-1}=\frac{1}{k}$, $k\ge 1$ with $r_0:=0$. Let also $\rH:=\rH_{R,\gA}$ be the corresponding Schr\"odinger operator  in $L^2(\R^n)$. 
Then:
 \item  $(i)$\  \ \  $\rH=\rH^*$ if \  $\sum_{k=1}^\infty |\alpha_k|k^{-3}=\infty$;
\item $(ii)$ \ $\rH=\rH^*$ if \ $\alpha_k\leq -2\bigl(2k+1)+O(k^{-1})$;
\item $(iii)$ \ $\rH=\rH^*$ if \ $\alpha_k\geq -Ck^{-1},\ k\in \N,\ C\equiv const>0$;
\item $(iv)$ \ $n_\pm(\rH)=\infty$ if \
$\alpha_k=-2k-1+O(k^{-\varepsilon})$ with some $ \varepsilon>0$;
\item  $(v)$ \ \
$n_\pm(\rH)=\infty$ if \ $\alpha_k=-A\left(2k+1\right)+O(k^{-1})$, $A\in(0,2)$.
    \end{corollary}
    \begin{proof}
The proof follows in a straightforward manner from Theorem \ref{th:sa} and
\cite[Example 5.12 and Proposition 5.13]{KM_09}
    \end{proof}

\subsection{Lower-semiboundedness}

Firstly, observe that Theorem \ref{th_KM}$(ii)$ can be easily extended to the multidimensional case.
\begin{theorem}\label{th_sb}
The operator $\rH_{R,\gA}$ is lower semibounded if and only if the matrix $B_{R,\gA}$ is also lower semibounded.
\end{theorem}
\begin{proof}
By Theorem \ref{th_KM}$(ii)$, the operators $\rh_{R,\gA}^{(l)}$, $l\in\N_0$, are lower semibounded if and only if so is $B_{R,\gA}$. Moreover, $\rh_{R,\gA}^{(l)}\ge \rh_{R,\gA}^{(0)}$ since the potential $q^{(l)}(r)=\frac{l(n)}{r^2}$, $l(n):=\frac{(n-1)(n-3)}{4} +l(l+n-2) $, is positive on $\R_+$ if $l\in\N$ and $n\in\N$. Hence, the representation \eqref{eq:31.01}--\eqref{eq:31.03} completes the proof.
\end{proof}
Also, we obtain the following generalization of the Glazman--Povzner theorem.
\begin{theorem}\label{th:GP}
The operator $\rH_{R,\gA}$ is self-adjoint if it is lower semibounded.
\end{theorem}
\begin{proof}
Firstly, observe that the operator $\rH_{R,\gA}$ is lower semibounded if and only if so is the operator $\rh_{R,\gA}^{(0)}$. By Theorem \ref{th_III.1}, the operator $\rh_{R,\gA}^{(0)}$ is self-adjoint if it is lower semibounded. However, the operators $\rh_{R,\gA}^{(0)}$ and $\rH_{R,\gA}$ are self-adjoint simultaneously.
\end{proof}
\begin{corollary}\label{cor:52.01}
Assume that $\gA=\{\gA_n\}_{1}^\infty$ satisfies the following condition
\begin{equation}\label{I_brinck_B}
\sup_{r>0}\sum_{r_k\in[r,r+1]}|\gA_k^-| <+\infty,\qquad \gA_k^-=(\gA_k-|\gA_k|)/2.
\end{equation}
Then the operator $\rH_{R,\gA}$ is self-adjoint and lower semibounded.

If $\gA=\gA^-$, then condition \eqref{I_brinck_B} is also necessary for $\rH_{R,\gA}$ to be lower semibounded.
    \end{corollary}
    \begin{proof}
    The claim follows by combining Theorem \ref{th_sb} with Theorem \ref{th_III.1} (see also \cite[Prop. 3.6]{AKM_10}).
    \end{proof}

\subsection{Characterization of the spectrum}

Before proceeding further, we need one result on the essential spectra of Hamiltonians with spherically symmetric potentials.
\begin{theorem}[\cite{hhk_87}]\label{th:hempel}
Let $\rH_{R,\gA}$ and $\rh_{R,\gA}^{(0)}$ be the operators \eqref{I_01} and \eqref{III_02}, respectively. Let also $R$ and $\gA$ satisfy \eqref{I_brinck_B}. Then $\sigma(\rH_{R,\gA})$ is bounded from below and
\[
\sigma_{\ess}(\rH_{R,\gA})=\big[\inf \sigma_{\ess}(\rh_{R,\gA}^{(0)}),+\infty\big).
\]
\end{theorem}
\begin{remark}
Theorem \ref{th:hempel} was obtained in \cite{hhk_87} in the case of locally integrable spherically symmetric potentials (see \cite[Theorem 2]{hhk_87}). However, the approach used there can be extended to the case of Hamiltonians with concentric $\delta$-shells (see, e.g., \cite{ef_07,ef_08}).
\end{remark}

\subsubsection{Discreteness}

In the study of the discreteness problem for the multidimensional operator we will restrict ourselves to the lower semibounded case.
\begin{theorem}
Let  the sequences $R$ and $\gA$ satisfy \eqref{I_brinck_B} and let $d_k=r_k-r_{k-1}\to 0$ as $k\to\infty$. 
The spectrum $\sigma(\rH_{R,\gA})$ of the operator
$\rH_{R,\gA}$ is discrete if and only if for every
$\varepsilon>0$ condition \eqref{2.1} holds true.
\end{theorem}
\begin{proof}
Under the assumption \eqref{I_brinck_B}, By Theorem \ref{th_discretcriter}, $\sigma(\rh_{R,\gA})$ is purely discrete, $\sigma_{\ess}(\rh_{R,\gA})=\emptyset$, if and only if \eqref{2.1} holds true for every $\varepsilon>0$. Therefore, by Theorem \ref{th:hempel}, $\sigma_{\ess}(\rH_{R,\gA})=\emptyset$ if and only if \eqref{2.1} holds true for every $\varepsilon>0$.
\end{proof}

\begin{remark}
If $\sigma(\rh_{R,\gA}^{(0)})$ is discrete but non lower semibounded, then it is possible to construct $R$ and $\gA$ such that $\sigma_{\ess}(\rH_{R,\gA})=\R$. We shall treat this case in greater detail elsewhere.
\end{remark}
%

\subsubsection{Continuous spectrum}
\begin{theorem}
Let the sequence $\gA$ satisfy \eqref{I_brinck}.
          If 
          \begin{equation}\label{3.2A}
\lim_{n\to\infty}\sum_{r_k\in[n,n+1]}|\alpha_k|=0,
     \end{equation}
      then $\sigma_{\ess}(\rH_{R,\gA})=\sigma_{\ess}(\rh_{R,\gA}^{(l)})=\R_+$.
\end{theorem}
\begin{proof}
Observe that the operators $\rH_{R,\gA}$ and $\rh_{R,\gA}^{(l)}$ are lower semibounded due to condition \eqref{I_brinck_B}. Using the partial wave decomposition, by Theorem \ref{thContSpec} we obtain
\[
\sigma_{\ess}(\rh_{R,\gA}^{(l)})=\R_+.
\]
It remains to note that by Theorem \ref{th:hempel}, $\sigma_{\ess}(\rh_{R,\gA}^{(l)})\subseteq \sigma_{\ess}(\rH_{R,\gA})\subseteq\R_+$.
%
\end{proof}

\subsubsection{Number of bound states}
The following equality clearly follows from the partial wave decomposition,
\be\label{eq:bs_mul}
\kappa_-(\rH_{R,\gA})=\sum_{l=0}^\infty \dim(\mathcal{H}_l)\kappa_-(\rh_{R,\gA}^{(l)}),
\ee
where $\mathcal{H}_l$ is the eigenspace corresponding to the $l$-th eigenvalue $\varkappa_l=-l(l+n-2)$ of the Laplace--Beltrami operator on $L^2(S^{n-1})$. Note that $\dim(\mathcal{H}_l)=2l+1$ if $n=3$ and in the case $n=2$
\[
\dim(\mathcal{H}_l)=\begin{cases}
                          2, & l>0\\
                          1, & l=0
                              \end{cases}.
\]

Therefore, the results of Section \ref{sec:bstates} provide estimates for the operator $\rH_{R,\gA}$. For instance, using the Bargman bound \eqref{eq:barg}, we can easily obtain the following estimates:
\begin{itemize}
\item[(i)] The case $n=3$:
\be
\kappa_-(\rH_{R,\gA})\le\sum_{l=0}^\infty (2l+1)\floor{\frac{I_0}{2l+1}}\le\frac{\floor{I_0}(\floor{I_0}+1)}{2},\quad I_0=\sum_{k=1}^\infty|\gA_k^-|r_k^-
\ee
\item[(ii)] The case $n=2$:
\be
\kappa_-(\rH_{R,\gA})\le \floor{I_{-1/2}}+\sum_{l=1}^\infty 2 \floor{\frac{I_0}{2l}}\le \floor{I_{-1/2}}+\floor{I_{0}}\log(\floor{I_0}),
\ee
where
\[
 I_{-1/2}=\sum_{k=1}^\infty|\gA_k^-\, r_k^-\log(r_k^-)|
\]

\end{itemize}


\section*{Appendix}

\appendix

\section{Boundary triplets and Weyl functions}
\label{app:hi}

Let $A$ be a densely defined closed symmetric operator in a
separable Hilbert space $\gH$ with equal deficiency indices
$\mathrm{n}_\pm(A)=\dim \gN_{\pm \I} \leq \infty$, $ \gN_z:=\ker(A^*-z)$.

\begin{definition}[\cite{Gor84}]\label{def:bt}
A triplet $\Pi=\{\cH,\gG_0,\gG_1\}$ is called a {\rm boundary
triplet} for the adjoint operator $A^*$ if $\cH$ is a Hilbert
space and $\Gamma_0,\Gamma_1:\  \dom(A^*)\rightarrow \cH$ are
bounded linear mappings such that the abstract Green identity
\begin{equation}\label{eq:green_f}
(A^*f,g)_\gH - (f,A^*g)_\gH = (\gG_1f,\gG_0g)_\cH - (\gG_0f,\gG_1g)_\cH, \quad
f,g\in\dom(A^*),
\end{equation}
holds
and the mapping $\gG:=\{\Gamma_0,\Gamma_1\}:  \dom(A^*)
\rightarrow \cH \oplus \cH$ is surjective.
\end{definition}
A boundary triplet 
for $A^*$ exists since the deficiency indices of $A$ are assumed to be
equal. Moreover, $A=A^*\upharpoonright\left(\ker(\Gamma_0) \cap \ker(\Gamma_1)\right)$ and $\mathrm{n}_\pm(A) = \dim(\cH)$ hold. Note also that a boundary triplet for $A^*$ is not unique. With every boundary triplet one naturally associates two self-adjoint extensions $A_0$ and $A_1$ of $A$ defined by
\be
A_j:=A^*\upharpoonright \ker (\Gamma_j),\quad j\in\{1,2\}.
\ee

In \cite{DM91}  the concept of the classical Weyl--Titchmarsh $m$-function
from the theory of Sturm-Liouville operators  was generalized to
the case of symmetric operators  with equal deficiency indices.
The role of abstract Weyl functions in the extension theory  is
similar to that of the classical Weyl--Titchmarsh $m$-function in the spectral
theory of singular Sturm-Liouville operators.

\begin{definition}[{\cite{DM91}}]\label{def:wf}
Let $A$ be a densely defined closed symmetric operator in $\gH$
with equal deficiency indices and let $\Pi=\{\cH,\gG_0,\gG_1\}$ be
a boundary triplet for $A^*$.
The operator valued function $M:\rho(A_0)\rightarrow  [\cH]$ defined by
\begin{equation}\label{II.1.3_01}
M(z):=\Gamma_1(\Gamma_0\upharpoonright\N_z)^{-1}, \qquad
z\in\rho(A_0),
\end{equation}
is called the 
{\em Weyl function} corresponding to the boundary triplet $\Pi$.
\end{definition}
The Weyl function
$M(\cdot)$ in \eqref{II.1.3_01}
is well defined. 
Moreover, $M(\cdot)$ is holomorphic on
$\rho(A_0)$ and is a Nevanlinna--Herglotz function  (see \cite{DM91}),
     \begin{equation}\label{II.1.3_03}
 \im z\cdot\im M(z)\geq 0,\qquad  M(z)^*=M(\overline
z),\qquad  z\in \C\setminus \R.
\end{equation}
Finally, we need the following result describing the number of negative squares of self-adjoint extensions of $A$ (see \cite[Section 4.4]{DM91}).

\begin{proposition}[\cite{DM91}]\label{prop:a.1}
Let $A$ be a nonnegative densely defined closed symmetric operator in $\gH$. Let $\Pi=\{\cH,\gG_0,\gG_1\}$ be
a boundary triplet for $A^*$ such that the operator $A_0$ is nonnegative. Let also $M(.)$ be the corresponding Weyl function. Then for any $B=B^*\in \mathcal{C}(\cH)$ the following equality holds true
\be
\kappa_-(A_B)=\kappa_-(B-M(0))-\kappa_-(B-M(-\infty)).
\ee
Here
\be
A_B:=A^*\upharpoonright \ker (\Gamma_1-B\Gamma_0),
\ee
and $M(0)$, $M(-\infty)$ are the strong resolvent limits (see, e.g., \cite[Chapter VIII]{Kato66})
\be
M(0):=s-R-\lim_{\lambda\uparrow 0}M(\lambda),\quad M(-\infty):=s-R-\lim_{\lambda\downarrow -\infty} M(\lambda).
\ee
\end{proposition}





\begin{thebibliography}{43}


\bibitem{as}
M. Abramovitz and I. A. Stegun, {\em Handbook of Mathematical
Functions}, Dover, New York, 1972.



\bibitem{Alb_Ges_88}
S. Albeverio, F. Gesztesy, R. Hoegh-Krohn, and H. Holden,
\emph{Solvable Models in Quantum Mechanics}, 2nd Edn. with an appendix by P. Exner, AMS Chelsea Publ., Providence, RI, 2005.

\bibitem{AKM_10}
S. Albeverio, A. Kostenko, and M. Malamud,
{\em Spectral theory of semi-bounded  Sturm-Liouville  operators with local interactions on a discrete set}, J. Math. Phys. \textbf{51}, 102102 (2010), 24 pp.


\bibitem{AlbNiz03}
S. Albeverio and L. Nizhnik,  {\em On the number of negative
eigenvalues of one--dimensional Schr\"odinger operator with point
interactions}, Lett. Math. Phys. \textbf{65},  27--35 (2003).

\bibitem{ags_87}
J.-P. Antoine, F. Gesztesy, and J. Shabani, {\em Exactly solvable models of sphere interactions in quantum mechanics}, J. Phys. A: Math. Gen. {\bf 20}, 3687--3712 (1987).

\bibitem{Ber68}
 Ju. M. Berezanskii,
 {\em Expansions in Eigenfunctions of Selfadjoint Operators},
 Transl. Math. Monographs, \textbf{17}, AMS, Providence, R.I., 1968.

\bibitem{Bir61}
M. S. Birman, {\em On spectrum of singular differential operators},
Math. Sbornik  \textbf{55} (2), 125--173 (1961).

\bibitem{bli_78}
S. M. Blinder, {\em Modified delta--function potential for hyperfine interactions}, Phys. Rev. A {\bf 18}, 853--861 (1978).

\bibitem{Bra85}
J. F. Brasche, {\em Perturbation of Schr\"odinger Hamiltonians by measures --- selfadjointness and semiboundedness},
J. Math. Phys. \textbf{26}, 621--626 (1985).

\bibitem{bg85}
W. Bulla and F. Gesztesy, {\em Deficiency indices and singular boundary conditions in quantum
mechanics}, J. Math. Phys. {\bf 26:10} (1985), 2520--2528.

\bibitem{DM91}
V. A. Derkach and M. M. Malamud, {\em Generalized resolvents and
the boundary value problems for Hermitian Operators with gaps},
J. Funct. Anal. \textbf{95}, 1--95 (1991).



\bibitem{ek_87}
Yu. V. Egorov and V. A. Kondratjev, {\em On an estimate of the number of points of the negative spectrum of the Schr\"odinger operator}, Mat. Sb. {\bf 134}, 556--570; {\em English transl. in} Sbornik Math. {\bf 62}, 551--566 (1989).

\bibitem{ef_07}
P. Exner and M. Fraas, {\em On the dense point and absolutely continuous
spectrum for Hamiltonians with concentric $\delta$--shells}, Lett. Math. Phys. {\bf 82}, 25--37 (2007).

\bibitem{ef_08}
P. Exner and M. Fraas, {\em Interlaced dense point and absolutely continuous
spectrum for Hamiltonians with concentric--shell singular interactions}, Proc.
QMath10 Conf. (Moeciu 2007; I. Beltita et al.,  eds.),
World Scientific, Singapore 2008; pp. 48--65.

\bibitem{Gla65}
I. M. Glazman, {\em Direct Methods of Qualitative
Spectral Analysis of Singular Differential Operators}, Fizmatgiz, Moscow,
1963.

\bibitem{gk}
I. C. Gohberg and M. G. Krein, {\em Introduction to the Theory of Linear Nonselfadjoint Operators},
Transl. Math. Monographs {\bf 18}, Amer. Math. Soc., Providence, RI, 1969.

\bibitem{GolOr_10}
N. I. Goloshchapova and L. L. Oridoroga, {\em On the negative spectrum  of  one-dimensional
Schr\"odinger operators  with point interactions}, Int. Equat. Oper. Theory {\bf 67} (1), 1--14
(2010).

\bibitem{Gor84}
V. I. Gorbachuk and M. L. Gorbachuk, {\em  Boundary Value Problems for Operator Differential Equations}, Mathematics and its
Applications (Soviet Series) 48, Kluwer Academic Publishers Group, Dordrecht, 1991.

\bibitem{gm_65}
I. M. Green and S. A. Moszkovski, {\em Nuclear coupling schemes with a surface delta interaction}, Phys. Rev. {\bf 134}, no.~4B, 790--793 (1965).

\bibitem{hhk_87}
R. Hempel, A. M. Hinz, and H. Kalf, {\em On the essential spectrum of Schr\"odinger operators with spherically symmetric potentials}, Math. Ann. {\bf 277}, 197--208 (1987).

%
%
\bibitem{KK58}
 I. S. Kac and M. G. Krein, {\em A discreteness criterion for the spectrum of a singular string}, Izvestiya Vuzov, Matematika, \textbf{3}(2), 136--153 (1958) (in Russian).
%

\bibitem{Kato66}
T. Kato, \emph{Perturbation Theory for Linear Operators}, 2nd Edn.,
Springer-Verlag, Berlin-Heidelberg, New York, 1966.



\bibitem{KM_09}
A. Kostenko and M. Malamud,  {\em 1--D Schr\"odinger operators with local point interactions
on a discrete set}, J. Differential Equations {\bf 249}, 253--304 (2010).

\bibitem{KM_2_09}
A. Kostenko and M. Malamud,  {\em One dimensional Schr\"odinger
operator with $\delta$-interactions}, Funct. Anal. Appl.
{\bf 44}(2), 87--91 (2010).

\bibitem{LanTis85}
P. Lancaster and M. Tismenetsky, {The Theory of Matrices. With Applications, Second Edition}, Academic Press, 1985.

\bibitem{llo_65}
P. Lloyd, {\em Pseudo-potential models in the theory of band structure}, Proc. Phys. Soc. {\bf 86}, 825--832 (1965).


\bibitem{MMM_92}
M. M. Malamud, {\em
On a formula of the generalized resolvents of a nondensely defined Hermitian operator}, Ukrain. Math. J. \textbf{44}(12), 1522--1547 (1992).



\bibitem{Ogu08}
O. Ogurisu, {\em On the number of negative eigenvalues of a
Schr\"odinger operator with point interactions}, Lett.
Math. Phys. \textbf{85}, 129--133 (2008).

\bibitem{Ogu10}
O. Ogurisu, {\em On the number of negative
eigenvalues of a Schr\"odinger operator with
$\delta$-interactions}, Methods Func. Anal. Topology
\textbf{16} (1), 42--50 (2010).









\bibitem{RedSim78}
M. Reed and B. Simon,  \emph{Methods of Modern Mathematical Physics, \textbf{IV:} Analysis of operators}, Academic Press,
New York, 1978.

\bibitem{rgm_67}
J. Rubio and F. Garcia--Moliner, {\em Formal theory of equivalent potentials in solids: II. Scattering theory approach for muffin--tin
potentials}, Proc. Phys. Soc. {\bf 92}, 206--214 (1967).


\bibitem{sha_88}
J. Shabani, {\em Finitely many $\delta$ interactions with supports on concentric spheres}, J. Math. Phys. {\bf 29}, 660--664 (1988).


\bibitem{Sim76}
B. Simon, {\em On the number of bound states of two body Schr\"odinger operators --- a review}, in: Studies in Math. Phys.: Essays in honor of V. Bargmann, by Lieb E. et. al. eds., Princeton, New Jersey, 1976; pp. 305--326.

\bibitem{Sim05}
B. Simon, {\em Trace Ideals and Their Applications: Second Edition}, Math. Surv. and Mon. {\bf 120}, Amer. Math. Soc., Rhode Island, 2005.


\end{thebibliography}
\end{document}